\newcommand{\extended}[1]{}    
\newcommand{\short}[1]{#1}     

\documentclass{aamas2017-arxiv}

\usepackage{amsfonts,amssymb,amsmath}
\usepackage{graphicx}
\usepackage{subfigure}
\usepackage{textcomp}
\usepackage[usenames]{color}
\usepackage{pgfplots}
\usepackage{multirow}
\usepackage{multicol}
\usepackage{algorithm}
\usepackage{algorithmic}
\usepackage{stmaryrd}
\usepackage{fixltx2e}
\usepackage{xspace}
\usepackage{url}
\usepackage{array}

\usepackage{threeparttable}
\usepackage{latexsym}
\usepackage{caption}
\usepackage{tikz}
\usetikzlibrary{arrows,automata,calc,shapes,snakes,backgrounds,petri,mindmap,fit,positioning}
\usepackage{amssymb,amsmath}
\usepackage{dblfloatfix}
\usepackage{fixltx2e}

\usepackage{wojtek15other}
\usepackage{wojtek15logics}
\usepackage{theoremsdecl}

\definecolor{lightyellow}{cmyk}{0,0,0.2,0}

\newcommand{\AMC}{\lan{AMC}}
\newcommand{\AEMC}{\lan{AE\ensuremath{\mu}C}}
\newcommand{\AMCi}{\AEMC}
\newcommand{\afAMCi}{\lan{af\text{-}AE\ensuremath{\mu}C}}
\newcommand{\diam}[1]{\ensuremath{\langle #1 \rangle}}
\newcommand{\diamPers}[1]{\ensuremath{\langle #1 \rangle^\circ}}
\newcommand{\diamStead}[1]{\ensuremath{\langle #1 \rangle^\bullet}}
\newcommand{\prot}{d}
\newcommand{\outf}{o}

\newcommand{\outir}{\mathit{out}^{\ir}}
\newcommand{\strats}{\Sigma}

\newcommand{\xvars}{\mathcal{V}\!\mathit{ars}}
\newcommand{\xvarsvals}{\mathcal{V}\!\mathit{als}}
\newcommand{\knowmod}{{K}}
\newcommand{\everyknowmod}{{E}}
\newcommand{\commonknowmod}{{C}}
\newcommand{\xval}{\mathcal{V}}

\newcommand{\sink}{\mathit{sink}}
\newcommand{\qprop}{\bar{Q}}
\newcommand{\Reach}{\mathit{Reach}}

\newcommand{\trlow}{tr}
\newcommand{\trup}{\mathit{TR}}
\newcommand{\Mvote}{M_{vote}}
\newcommand{\onlabel}[1]{\colorbox{white}{{{#1}}}}

\newcommand{\bend}{20}
\newcommand{\bendtwo}{25}

\definecolor{lightgrey}{rgb}{0.8,0.8,0.8}
\definecolor{grey}{rgb}{0.6,0.6,0.6}
\definecolor{darkgrey}{rgb}{0.4,0.4,0.4}
\definecolor{darkgreen}{rgb}{0,0.6,0}

\renewcommand{\ICGS}{CEGS\xspace}
\renewcommand{\phi}{\varphi\xspace}
\newcommand{\tmout}{{\small timeout}}

\newcommand{\South}{\textbf{S}\xspace}
\newcommand{\West}{\textbf{W}\xspace}
\newcommand{\North}{\textbf{N}\xspace}
\newcommand{\East}{\textbf{E}\xspace}

\newcolumntype{H}{>{\setbox0=\hbox\bgroup}c<{\egroup}@{}}

\newcommand{\wj}[1]{\textcolor{darkgreen}{\textbf{WJ:} #1\ }}
\renewcommand{\wj}[1]{}

\begin{document}

\title{Fixpoint Approximation of Strategic Abilities under Imperfect Information}

\numberofauthors{3}
\author{
  Wojciech Jamroga\\
  \affaddr{Institute of Computer Science,}\\
  \affaddr{Polish Academy of Sciences}\\
  \email{w.jamroga@ipipan.waw.pl}
\alignauthor
  Micha\l\ Knapik\\
  \affaddr{Institute of Computer Science,}\\
  \affaddr{Polish Academy of Sciences}\\
  \email{michal.knapik@ipipan.waw.pl}
\alignauthor
  Damian Kurpiewski\\
  \affaddr{Institute of Computer Science,}\\
  \affaddr{Polish Academy of Sciences}\\
  \email{damian.kurpiewski@ipipan.waw.pl}
}

\maketitle

\begin{abstract}
Model checking of strategic ability under imperfect information is known to be hard.
The complexity results range from \NP-completeness to undecidability, depending on the precise setup of the problem. No less importantly, fixpoint equivalences do not generally hold for imperfect information strategies, which seriously hampers incremental synthesis of winning strategies.

In this paper, we propose translations of \ATLir formulae that provide lower and upper bounds for their truth values, and are cheaper to verify than the original specifications.
That is, if the expression is verified as true then the corresponding formula of \ATLir should also hold in the given model.
We begin by showing where the straightforward approach does not work. Then, we propose how it can be modified to obtain guaranteed lower bounds. To this end, we alter the next-step operator in such a way that traversing one's indistinguishability relation is seen as atomic activity.
Most interestingly, the lower approximation is provided by a fixpoint expression that uses a nonstandard variant of the next-step ability operator.
We show the correctness of the translations, establish their computational complexity, and validate the approach by experiments with a scalable scenario of Bridge play.

\medskip\noindent
\textbf{The paper will appear in: \emph{S. Das, E. Durfee, K. Larson, M. Winikoff (eds.), Proceedings of the 16th International Conference on Autonomous Agents and Multiagent Systems (AAMAS 2017),
May 8--12, 2017, Sao Paulo, Brazil}}.
\end{abstract}

%

\printccsdesc




\section{Introduction}\label{secIntro}

\extended{
  Multi-agent systems describe interactions of multiple entities called \emph{agents}, often assumed to be intelligent
  and autonomous.
  More and more practical problems are being modeled and solved under paradigms related to multi-agent systems.
  The examples of their industrial applications include space mission planning and air control\cite{ClanceySSBRHS07,Rouff2007,Tumer2007},
  defense and security~\cite{Gascuena2011,Pechoucek2008},
  logistics and production planning~\cite{Himoff2005,Jacobi2005}, and many others.
  According to the 2008 review~\cite{Pechoucek2008rev}, companies like Rockwell Automation, Volkswagen,
  SkodaAuto, Daimler AG, NASA, and JPL either routinely use or are actively developing multi-agent
  based solutions.
  Given the proliferation of multi-agent systems in the modern world, their automated verification, i.e., ensuring of
  the compliance with their specification, becomes an important task.
  Some model checking tools (e.g., MCK~\cite{mck}, MCMAS~\cite{LomuscioQR09})
  have been created or extended to accept multi-agent systems as their input models;
  some tools aim at the verification of programs specified in multi-agent programming
  languages (e.g., AgentSpeak in Jason~\cite{Bordini03a}).

  In this paper we deal with the properties specified in Alternating-Time Temporal Logic~\cite{AlurHK02} (\ATL).
  In its basic version, this logic allows to specify strategic properties of groups of agents
  under assumption of the perfect knowledge about the state of affairs, i.e., an agent is able to
  observe the world in its entirety. One can argue  that such a scenario is rather
  unrealistic, hence several versions of semantics have been introduced~\cite{Schobbens04}
  to take into account the limited knowledge of the agents. This, however, typically leads
  to an increase of the complexity of the task of model checking as compared to vanilla \ATL,
  where it can be done in linear time in the size of the model and the formula.
} 

There is a growing number of works that study the syntactic and semantic variants of the \emph{strategic logic} \ATL for agents with imperfect information~\cite{Agotnes15handbook}.
The contributions are mainly theoretical, and include results concerning the conceptual soundness of a given semantics\extended{ of ability}~\cite{Schobbens04ATL,Jamroga03FAMAS,Agotnes04atel,Jamroga04ATEL,Dima10communicating,Guelev12stratcontexts,Agotnes15handbook}, meta-logical properties~\cite{Guelev11atl-distrknowldge,Bulling14comparing-jaamas}, and the complexity of model checking~\cite{Schobbens04ATL,Jamroga06atlir-eumas\extended{,Jamroga06mis-tr},Guelev11atl-distrknowldge,Hoek06practicalmcheck,Dima11undecidable,Bulling10verification}.
However, there is relatively little research on the\extended{ actual} \emph{use} of the logics, in particular on practical algorithms for reasoning and/or verification in scenarios where agents have a limited view of the world.

This is somewhat easy to understand, since model checking of \ATL variants with imperfect information has been proved \Deltwo- to \Pspace-complete for agents playing memoryless\extended{ (a.k.a.~positional)} strategies~\cite{Schobbens04ATL,Jamroga06atlir-eumas,Bulling10verification} and \EXPTIME-complete to undecidable for agents with perfect recall of the past~\cite{Dima11undecidable,Guelev11atl-distrknowldge}.
Moreover, the imperfect information semantics of \ATL does not admit alternation-free fixpoint characterizations~\cite{Bulling11mu-ijcai,Dima14mucalc,Dima15fallmu}, which makes incremental synthesis of strategies impossible, or at least difficult to achieve.
Some early attempts at verification of imperfect information strategies made their way into the MCMAS model-checker~\cite{Lomuscio06uniform,Raimondi06phd,Lomuscio09mcmas,Lomuscio15mcmas}, but the issue was never at the heart of the tool.
More dedicated attempts began to emerge only recently~\cite{Pilecki14synthesis,Busard14improving,Huang14symbolic-epist,Busard15reasoning}.
Up until now, experimental results confirm that the initial intuition was right: model checking of strategic modalities for imperfect information is hard, and dealing with it requires innovative algorithms and verification techniques.

In this paper, we propose that in some instances, instead of the exact model checking, it suffices to provide an upper and/or lower bound for the output.
The intuition for the upper bound is straightforward: instead of checking existence of an imperfect information strategy, we can look for a perfect information strategy that obtains the same goal. If the latter is false, the former must be false too.
Finding a reasonable lower bound is nontrivial, but we construct one by means of a fixpoint expression in alternating epistemic mu-calculus.
We begin by showing that the straightforward fixpoint approach does not work. Then, we propose how it can be modified to obtain guaranteed lower bounds. To this end, we alter the next-step operator in such a way that traversing the appropriate epistemic neighborhood is seen as an atomic activity.
We show the correctness of the translations, establish their computational complexity, and validate the approach by experiments with some scalable scenarios.

\extended{
  One idea that has not been properly explored is that of alternating-time epistemic mu-calculus (\AEMC)~\cite{Bulling11mu-ijcai}. Since fixpoint equivalences do not hold under imperfect information, it follows that standard fixpoint translations of \ATL modalities lead to a different interpretation of strategic ability. In fact, it can be argued that they capture existence of \emph{recomputable} winning strategies. However, what especially interests us in the context of model checking is that they can make model checking computationally cheaper. Verification of \AEMC is in general between \NP and \Deltwo, but the scope of backtracking is much smaller than for \ATL with imperfect information, as it includes only the actions starting from a given indistinguishability class rather than all the actions in the model. Moreover, for coalitions of up to 2 agents the model checking problem is in \Ptime~\cite{Bulling11mu-ijcai}.

  =====

  As shown in~\cite{Bulling11mu-ijcai}, there is no clear correspondence between the validity of the formulae
  specified in \ATLir and their natural fixed-point counterparts specified
  in Alternating Epistemic Mu-Calculus (\AMCi, for short).
  This is unlike in \ATL that can be embedded into Alternating Epistemic Mu-Calculus
  (\AMC).
  We analyse the reasons for this lack of correspondence and
  propose a novel method for underapproximating \ATLir by means of fixed-point calculations.
  To this end we introduce a special version of the next-step operator,
  called Persistent Imperfect Next-Step Operator $\diamStead{\cdot}$ and
  show how it can be used to define a new version of reachability that carries to \ATLir.
  We extend the intuitions built during the analysis of the reachability
  to a wider set of formulae and define a new logic,
  called Approximative \ATLir, that underapproximates \ATLir.

  To the best knowledge of the authors this is the first successful attempt at underapproximating \ATLir using fixed-point methods.
} 

\extended{
  \subsection{Related Work}

  Alternating-Time Temporal Logic and
  Alternating Mu-Calculus (\AMC) are introduced in~\cite{AlurHK02}.
  \ATL deals with strategic abilities of coalitions of agents,
  while \AMC
  combines the next-step operator of \ATL with the operator
  of the least fixed-point $\mu$. In~\cite{AlurHK02} the authors show that \AMC is
  strictly more expressive than the original \ATL.
  However, the semantics of \ATL is also considered in
  several flavours that rely upon the definition of strategy.
  In~\cite{Schobbens04} a natural classification of these versions is presented,
  depending on
  (1) whether the agents make a decision about the next step based on the
  history of the visited states or based on the current state;
  (2) whether the agents can observe the entire state of affairs or only their
  local epistemic neighbourhood.
  In this paper we focus on the version of \ATL with memoryless
  strategies and the agents that have a limited knowledge about the world,
  i.e., \ATLir in the nomenclature of~\cite{Schobbens04}.
  In~\cite{Bulling11mu-ijcai} Alternating Epistemic Mu-Calculus,
  a fixed-point logic built on top of the next-step operator of \ATLir
  is presented and investigated;
  it is shown that the expressivity of \ATLir is incomparable with the expressivity of \afAMCi.
}

\section{Verifying Strategic Ability}\label{secPrelims}


In this section we provide an overview of the relevant variants of \ATL\extended{, and the corresponding complexity results for model checking}. We refer the \extended{interested reader }to~\cite{Alur02ATL,Hoek02ATEL,Schobbens04ATL,Bulling11mu-ijcai,Jamroga15specificationMAS} for details.

\subsection{Models, Strategies, Outcomes}


\extended{We interpret specifications over a variant of transition systems where transitions are labeled with combinations of actions, one per agent. Moreover, epistemic relations are used to indicate states that look the same to a given agent.
Formally, a}%
\short{A} \emph{concurrent epistemic game structure} or \emph{\ICGS}\extended{~\cite{Alur02ATL,Hoek02ATEL}} is given by
$\model = \tuple{\Agt, \States, \Props, \V, Act, d, o,\{\sim_a \mid a\in \Agt\}}$
which includes a nonempty finite set of all agents $\Agt = \set{1,\dots,k}$, a nonempty set of states $\States$, a set of atomic propositions $\Props$ and their valuation $\V:\Props\rightarrow \powerset{\States}$, and a nonempty finite set of (atomic) actions $\Actions$. Function $\prot : \Agt \times \States \rightarrow \powerset{Act}$ defines nonempty sets of actions available to agents at each state, and $o$ is a (deterministic) transition function that assigns the outcome state $q' = o(q,\alpha_1,\dots,\alpha_k)$ to state $q$ and a tuple of actions $\langle\alpha_1, \dots, \alpha_k\rangle$\extended{ for $\alpha_i \in d(i,q)$ and $1\leq i\leq k$,} that can be executed by $\Agt$ in $q$.
We write $\prot_a(q)$ instead of $\prot(a,q)$\extended{, and define $\prot_A(q) = \prod_{a\in A}\prot_a(q)$ for each $A\subseteq\Agt, q\in\States$}.
Every $\sim_a\subseteq \States\times\States$ is an epistemic equivalence relation.
The \ICGS is assumed to be \emph{uniform}, in the sense that $q\sim_a q'$ implies $d_a(q)=d_a(q')$.
\extended{Note that perfect information can be modeled by assuming each $\sim_a$ to be the minimal reflexive relation.}

\begin{figure}[h]
\hspace{-0.5cm}
\begin{tikzpicture}[>=latex,scale=1.15]
 \input{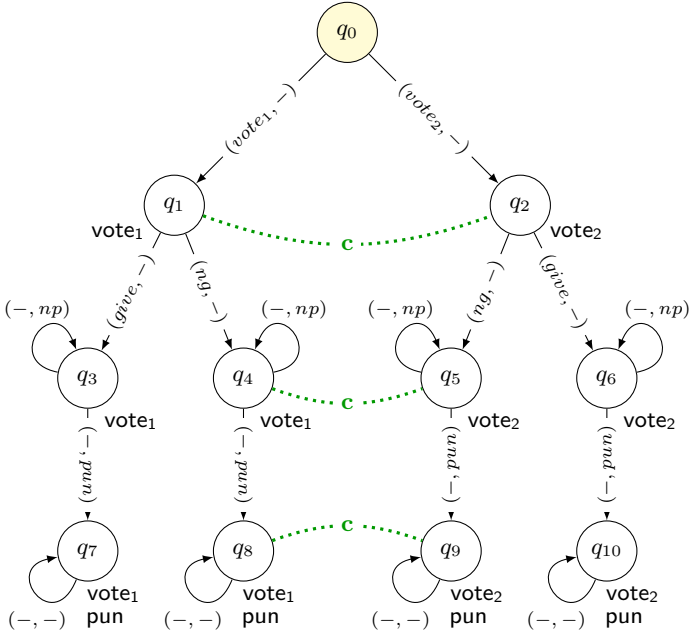}
\end{tikzpicture}
\caption{A simple model of voting and coercion}
\label{fig:votingmodel}
\end{figure}

\begin{example}\label{ex:voting-model}
Consider a very simple voting scenario with two agents: the voter $v$ and the coercer $c$. The voter casts a vote for a selected candidate $i\in\set{1,\dots,n}$ (action $vote_i$).
Upon exit from the polling station, the voter can hand in a proof of how she voted to the coercer (action $give$) or refuse to hand in the proof (action $ng$). The proof may be a certified receipt from the election authorities, a picture of the ballot taken with a smartphone, etc.\extended{ -- anything that the coercer will consider believable.}
After that, the coercer can either punish the voter (\extended{action }$pun$) or not punish (\extended{action }$np$).

The \ICGS\ $\Mvote$ modeling the scenario for $n=2$ is shown in Figure~\ref{fig:votingmodel}.
Proposition $\prop{vote_i}$ labels states where the voter has already voted for candidate $i$.
Proposition $\prop{pun}$ indicates states where \short{$v$}\extended{the voter} has been punished.
The indistinguishability relation for the coercer is depicted by dotted lines.
\end{example}

A \emph{strategy} of agent $a\in\Agt$ is a conditional plan that specifies what $a$ is going to do in every possible situation.
Formally, a \emph{perfect information memoryless strategy} for $a$ can be represented by a function $s_a : \States\to\Actions$ satisfying $s_a(q)\in\prot_a(q)$ for each $q\in\States$.
An \emph{imperfect information memoryless strategy} additionally satisfies that $s_a(q) = s_a(q')$ whenever $q\sim_a q'$.
Following~\cite{Schobbens04ATL}, we refer to the former as \emph{\Ir-strategies}, and to the latter as \emph{\ir-strategies}.

A \emph{collective $x$-strategy} $s_A$, for \extended{coalition }$A\subseteq\Agt$ and \extended{strategy type }$x\in\set{\Ir,\ir}$, is a tuple of individual $x$-strategies, one per agent from $A$. The set of all such strategies is denoted by $\strats_A^x$.
By $s_A|_a$ we denote the strategy of agent $a\in A$ selected from $s_A$.

Given two partial functions $f,f' \colon X\fpart Y$, we say that \emph{$f'$ extends $f$} (denoted $f\subseteq f'$) if, whenever $f(x)$ is defined, we have $f(x) = f'(x)$.
A partial function $s_a'\colon\States\fpart\Actions$ is called \emph{a partial $x$-strategy for $a$} if $s_a'$ is extended by some strategy $s_a\in\strats_a^x$.
A collective partial x-strategy $s_A$ is a tuple of partial x-strategies, one per agent from $A$.

A \emph{path} $\lambda=q_0q_1q_2\dots$ is an infinite sequence of
states such that there is a transition between each $q_i,q_{i+1}$.
We use $\lambda[i]$ to denote the $i$th position on path $\lambda$
(starting from $i=0$)\extended{ and $\lambda[i,\infty]$ to denote the subpath
of $\lambda$ starting from $i$}.
%
Function $out(q,s_A)$ returns the set of all paths that can result from the execution of strategy $s_A$ from state $q$\short{. }\extended{, defined formally as follows:
\begin{description}
\item[$out(q,s_A) =$]  $\{ \lambda=q_0,q_1,q_2\ldots \mid
      q_0=q$ and for each $i=0,1,\ldots$ there exists
      $\tuple{\alpha^{i}_{a_1},\ldots,\alpha^{i}_{a_k}}$ such that
      $\alpha^{i}_{a} \in d_a(q_{i})$ for every $a\in \Agt$,
      and $\alpha^{i}_{a} = s_A[a](q_{i})$ for every $a\in A$,
      and $q_{i+1} = o(q_{i},\alpha^{i}_{a_1},\ldots,\alpha^{i}_{a_k}) \}$.
\end{description}
}
We will sometimes write $out^\Ir(q,s_A)$ instead of $out(q,s_A)$.
Moreover, function $out^{\ir}(q,s_A) = \bigcup_{a\in A}\bigcup_{q\sim_a q'}\out(q',s_A)$ collects all the outcome paths that start from states that are indistinguishable from $q$ to at least one agent in $A$.

\subsection{Alternating-Time Temporal Logic}

We use a variant of \ATL that explicitly distinguishes between perfect and imperfect information abilities.
Formally, the syntax is defined by the following grammar:
\[
\phi ::= p \mid \neg\phi \mid \phi\land\phi \mid \coop[x]{A}\Next\phi
\mid \coop[x]{A}\Always\phi \mid \coop[x]{A}\phi \Until\phi,
\]
where $x\in\set{\Ir,\ir}$, $p \in \Props$ and $A\subseteq\Agt$.
We read $\coop[\ir]{A}\gamma$ as \emph{``$A$ can identify and execute a strategy that enforces $\gamma$,''}
$\Next$ as \emph{``in the next state,''} $\Always$ as \emph{``now and always in the future,''} and $\Until$ as \emph{``until.''}
$\coop[\Ir]{A}\gamma$ can be read as \emph{``$A$ might be able to bring about $\gamma$ if allowed to make lucky guesses along the way.''}
We focus on the kind of ability expressed by $\coop[\ir]{A}$. The other strategic modality (i.e., $\coop[\Ir]{A}$) will prove useful when approximating $\coop[\ir]{A}$.

The semantics of \ATL can be defined as follows:
\begin{itemize2}
\item $\model,q\satisf p$ iff $q\in\V(p)$,
\item $\model,q\satisf\neg\phi$ iff $\model,q\not\satisf\phi$,
\item $\model,q\satisf\phi\land\psi$ iff $\model,q\satisf\phi$ and $\model,q\satisf\psi$,
\item $\model,q\satisf\coop[x]{A}\Next\phi$ iff there exists $s_A\in\strats_A^x$ such that for all $\lambda\in\out^x(q,s_A)$ we have $\model,\lambda[1]\satisf\phi$,
\item $\model,q\satisf\coop[x]{A}\Always\phi$ iff there exists $s_A\in\strats_A^x$ such that for all $\lambda\in\out^x(q,s_A)$
    and $i\in\Nat$ we have $\model,\lambda[i]\satisf\phi$,
\item $\model,q\satisf\coop[x]{A}\psi\Until\phi$ iff there exists $s_A\in\strats_A^x$ such that for all $\lambda\in\out^x(q,s_A)$ there is $i\in\Nat$ for which $\model,\lambda[i]\satisf\phi$ and $\model,\lambda[j]\satisf\psi$ for all $0\le j < i$.
\end{itemize2}
We will often write $\diam{A}\varphi$ instead of $\coop[\ir]{A}\Next\varphi$ to express one-step abilities under imperfect information.
Additionally, we define \emph{``now or sometime in the future''} as $\Sometm\varphi \equiv \top\Until\varphi$.
\extended{
  It is easy to see that $\model,q\satisf\coop[x]{A} \Sometm\phi$ iff there exists \extended{a collective strategy }$s_A\in\strats_A^x$ such that
  following each path $\lambda\in\out^x(q, s_A)$ leads to some state satisfying $\phi$.
  In that case, we say that $\phi$ is \emph{$x$-reachable from $q$}.
} 

\begin{example}\label{ex:voting-formulae}
Consider model $\Mvote$ from Example~\ref{ex:voting-model}. The following formula expresses that the coercer can ensure that the voter will eventually either have voted for candidate $i$ (presumably chosen by the coercer for the voter to vote for) or be punished:\
$\coop[\ir]{c}\Sometm \big(\neg\prop{pun} \then \prop{vote_i}\big)$.
We note that it holds in $\Mvote,q_0$ for any $i=1,2$.
A strategy for $c$ that validates the property is $s_c(q_3)=np,\ s_c(q_4)=s_c(q_5)=s_c(q_6)=pun$ for $i=1$, and symmetrically for $i=2$.

Consequently, the formula $\coop[\ir]{v}\Always \big(\neg\prop{pun} \land \neg\prop{vote_{i}}\big)$ saying that the voter can avoid voting for candidate $i$ and being punished, is false in $\Mvote,q_0$ for all $i=1,2$.
\end{example}

We refer to the syntactic fragment containing only $\coop[\ir]{A}$ modalities as \ATLir, and to the one containing only $\coop[\Ir]{A}$ modalities as \ATLIr.

\begin{proposition}[\cite{Alur02ATL,Schobbens04ATL,Jamroga06atlir-eumas}]\label{prop:complexity-atl}
Model checking \ATLIr is \Ptime-complete and can be done in time $O(|M|\cdot|\phi|)$ where $|M|$ is the number of transitions in the model and $|\phi|$ is the length of the formula.

Model checking \ATLir is \Deltwo-complete wrt $|M|$ and $|\phi|$.
\end{proposition}


\begin{remark}\label{rem:subjective}
The semantics of $\coop[\ir]{A}\gamma$ encodes the notion of ``subjective'' ability~\cite{Schobbens04ATL,Jamroga04ATEL}: the agents must have a successful strategy from all the states that they consider possible when the system is in state $q$. Then, they know that the strategy indeed obtains $\gamma$. The alternative notion of ``objective'' ability~\cite{Bulling14comparing-jaamas} requires\extended{ existence of} a winning strategy from state $q$ alone. We focus on the subjective interpretation, as it is more standard in \ATL and more relevant in game solving (think of a card game, such as poker or bridge: the challenge is to find a strategy that wins \emph{for all possible hands of the opponents}).

Note that if $[q]_{\sim_A^E} = \set{q}$ and $\gamma$ contains no nested strategic modalities, then the subjective and objective semantics of $\coop[\ir]{A}\gamma$ at $q$ coincide.
Moreover, model checking $\coop[\ir]{A}\prop{p_1}\Until\prop{p_2}$ and $\coop[\ir]{A}\Always\prop{p}$ in $M,q$ according to the objective semantics can be easily reduced to the subjective case by adding a spurious initial state $q'$, with transitions to all states in $[q]_{\sim_A^E}$, controlled by a ``dummy'' agent outside $A$~\cite{Pilecki17smc}.
\end{remark}

\subsection{Reasoning about Knowledge}

Having indistinguishability relations in the models, we can interpret knowledge modalities $K_a$ in the standard way:
\begin{itemize}
\item $\model,q\satisf K_a\phi$ iff $\model,q'\satisf\phi$ for all $q$ such that $q\sim_a q'$.
\end{itemize}
The semantics of \emph{``everybody knows'' ($E_A$)} and \emph{common knowledge ($C_A$)} is defined analogously by assuming the relation $\sim_A^E = \bigcup_{a\in A}\sim_a$ to aggregate individual uncertainty in $A$, and $\sim_A^C$ to be the transitive closure of $\sim_A^E$.
Additionally, we take $\sim_\emptyset^E$ to be the minimal reflexive relation.
We also use $[q]_\mathcal{R} = \set{q' \mid q\mathcal{R}q'}$ to denote the image of $q$ wrt relation $\mathcal{R}$.

\begin{example}\label{ex:voting-epist}
The following formulae hold in $\Mvote,q_0$ for any $i=1,2$ by virtue of strategy $s_c$ presented in Example~\ref{ex:voting-formulae}:
\begin{itemize2}
\item
$\coop[\ir]{c}\Sometm \big((\neg K_c\prop{vote_i}) \then \prop{pun}\big)$:
The coercer has a strategy so that, eventually, the voter is punished unless the coercer has learnt that the voter voted as instructed;

\item
$\coop[\ir]{c}\Always \big((K_c \prop{vote_i}) \then \neg\prop{pun}\big)$:
Moreover, the coercer can guarantee that if he learns that the voter obeyed, then the voter will not be punished.

\end{itemize2}
\end{example}

\subsection{Alternating Epistemic Mu-Calculus}\label{sec:aemc}

It is well known that the modalities in \ATLIr have simple fixpoint characterizations~\cite{Alur02ATL}, and hence \ATLIr can be embedded in a variant of $\mu$-calculus with $\coop[\Ir]{A}\Next$ as the basic modality.
At the same time, the analogous variant of $\mu$-calculus for imperfect information has incomparable expressive power to \ATLir~\cite{Bulling11mu-ijcai}, which suggests that, under imperfect information, \ATL and fixpoint specifications provide different views of strategic ability.

Formally, \emph{alternating epistemic $\mu$-calculus (\AMCi)} takes the next-time fragment of \ATLir, possibly with epistemic modalities, and adds the least fixpoint operator $\mu$. The greatest fixpoint operator $\nu$ is defined as dual.
Let  $\xvars$ be a set of second-order variables ranging over $\powerset{\States}$.
The language of \AMCi is defined by the following grammar:
\[
\phi ::= p \mid Z\mid \neg\phi \mid \phi\lor\phi \mid \langle A \rangle \phi
\mid \mu Z(\phi) \mid \knowmod_a,
\]
where $p \in \Props$, $Z\in\xvars$, $a\in\Agt$, $A\subseteq\Agt$,
and the formulae are $Z$--positive, i.e., each free occurrence of $Z$ is in the scope
of an even number of negations.
We define $\nu Z(\phi(Z))\equiv \neg\mu Z(\neg \phi (\neg Z))$.
A formula of \AMCi is \emph{alternation-free} if in its negation normal form
it contains no occurrences of $\nu$ (resp. $\mu$) on any syntactic path
from an occurrence of $\mu Z$ (resp. $\nu Z$) to a bound occurrence of $Z$.
\extended{
  Similarly to~\cite{Bulling11mu-ijcai}, we consider here only the alternation-free fragment of \AMCi, denoted by \afAMCi.

  We evaluate the formulae of \afAMCi with respect to the valuations of $\xvars$,
  i.e., functions $\xval\colon\xvars\to\powerset{\States}$. We denote the set of all the valuations
  of $\xvars$ by $\xvarsvals$.
  If $X\in\xvars$, $Z\subseteq\States$, and $\xval\in\xvarsvals$, then
  by $\xval[X := Z]$ we denote the valuation of $\xvars$ such that
  $\xval[X := Z](Y) = \xval(Y)$ for $Y\ne X$ and $\xval[X := Z](X) = Z$.
}

The denotational semantics of \afAMCi\ (i.e., the alterna\-tion-free fragment of \AMCi) assigns to each formula $\phi$ the set of states $\denotation{\phi}^\model_{\xval}$
where $\phi$ is true under the valuation $\xval\in\xvarsvals$:
\begin{itemize2}
\item $\denotation{p}^\model_{\xval} = \V(p)$,
\qquad\qquad $\denotation{Z}^\model_{\xval} = \xval(Z)$,
\item $\denotation{\neg\phi}^\model_{\xval} = \States\setminus\denotation{\phi}^\model_{\xval}$,
\item $\denotation{\phi\lor\psi}^\model_\xval = \denotation{\phi}^\model_{\xval}\cup\denotation{ \psi}^\model_{\xval}$,
\item $\denotation{\langle A \rangle\phi }^\model_{\xval} =
\{q\in\States \mid \exists s_A\in\strats_A\;\forall\lambda\in\outir_{\model}(q, s_A)\;\\ \lambda[1]\in \denotation{\phi}^\model_{\xval}\}$,
\item $\denotation{\mu Z(\phi)}^\model_{\xval} =
\bigcap\{Q\subseteq\States\mid \denotation{\phi}^\model_{\xval[Z := Q]}\subseteq Q\}$,
\item $\denotation{\knowmod_a \phi}^\model_{\xval} \!=\!
\{q\in\States \mid\forall_{q'}(q'\!\sim_a\! q \text{ implies } q
\!'\in\!\denotation{\phi}^\model_{\xval})\}$\short{.}\extended{,}
\end{itemize2}
\extended{where $\phi\in\afAMCi$, $p \in \Props$, $Z\in\xvars$, $A\subseteq\Agt$, and $a\in\Agt$.}
%
If $\phi$\extended{ is a sentence, i.e., it} contains no free variables, then its validity does not depend on\extended{ the valuation} $\xval$,
and we write $\model,q\models\phi$ instead of $q\in\denotation{\phi}^\model_{\xval}$.

\begin{example}\label{ex:voting-mu}
Consider the \AMCi formula $\mu Z. \big((\neg\prop{pun} \then \prop{vote_i}) \lor \diam{c}Z\big)$, i.e., the ``naive'' fixpoint translation of
the formula $\coop[\ir]{c}\Sometm \big(\neg\prop{pun} \then \prop{vote_i}\big)$ from Example~\ref{ex:voting-formulae}.
The fixpoint computation produces the whole set of states $\States$.
Thus, in particular, $\Mvote,q_0 \models \mu Z. \big((\neg\prop{pun} \then \prop{vote_i}) \lor \diam{c}Z\big)$.
\end{example}

\begin{proposition}[\cite{Bulling11mu-ijcai}]\label{prop:complexity-aemc}
Model checking \afAMCi with strategic modalities for up to 2 agents is \Ptime-complete and can be done in time $O(|\!\!\sim\!\!|\cdot|\phi|)$ where $|\!\!\sim\!\!|$ is the size of the largest equivalence class among $\sim_1,\dots,\sim_k$, and $|\phi|$ is the length of the formula.

For coalitions of size at least 3, the problem is between \NP and \Deltwo wrt $|\!\!\sim\!\!|$ and $|\phi|$.
\end{proposition}

Thus, alternation-free alternating epistemic $\mu$-calculus can be an attractive alternative to \ATLir from the complexity point of view.
Unfortunately, formulae of \ATLir admit no universal translations to \afAMCi.
Formally, it was proved in~\cite[Proposition~6]{Bulling11mu-ijcai} that \afAMCi does not cover the expressive power of \ATLir. The proof uses formulae of type $\coop{a}\Sometm\prop{p}$, but it is easy to construct an analogous argument for $\coop{a}\Always\prop{p}$.
In consequence, long-term strategic modalities of \ATLir do not have alternation-free fixpoint characterizations in terms of the next-step strategic modalities $\diam{A}$.
\extended{\wj{Can we extend the result and prove the same for whole \AMCi? I believe it should be true.}}
A similar result was proved for \ATLiR\extended{, i.e., the case of imperfect information and perfect recall strategies,} in~\cite[Theorem~11]{Dima14mucalc}.

\section{Lower Bounds for Abilities}\label{secUnder}


The complexity of \AEMC model checking seems more attractive than that of \ATLir. Unfortunately, the expressivity results cited in Section~\ref{sec:aemc} imply that there is no simple fixpoint translation which captures \emph{exactly} the meaning of \ATLir operators. It might be possible, however, to come up with a translation $\trlow$ that provides a \emph{lower bound} of the actual strategic abilities, i.e., such that $M,q \satisf \trlow(\coop[\ir]{A}\gamma)$ implies $M,q \satisf \coop[\ir]{A}\gamma$. In other words, a translation which can only reduce, but never enhance the abilities of the coalition.

We begin by investigating the ``naive'' fixpoint translation that mimics the one for \ATLIr, and show that it works in some cases, but not in general.
Then, we propose how to alter the semantics of the nexttime modality so that a general lower bound can be obtained.
We focus first on reachability goals, expressed by formulae $\coop[\ir]{A}\Sometm\phi$, and then \extended{extend the approach}\short{move on} to the other modalities.

%

\subsection{Trying It Simple for Reachability Goals}

We assume from now on that $\phi$ is a formula of \ATLir, $\model$ is a \ICGS, and $q$ is a state in $M$ (unless explicitly stated otherwise).
We start with the simplest translation, analogous to that of~\cite{Alur02ATL}:
$\trlow_1(\coop[\ir]{A}\Sometm\phi) = \mu Z.(\phi \lor \diam{A} Z)$.
Unfortunately, this translation provides neither a lower nor an upper bound.
For the former, use model $\model_0$ in Figure~\ref{fig:tr1}, and observe that $\model_0,q_0 \models \mu Z.(\prop{p} \lor \diam{1} Z)$ but $\model_0,q_0 \not\models \coop[\ir]{1}\Sometm\prop{p}$.
For the latter, take model $M$ in~\cite[Figure~1]{Bulling11mu-ijcai}, and observe that $M,q_0 \models \coop[\ir]{1}\Sometm\prop{p}$ but $M,q_0 \not\models \mu Z.(\prop{p} \lor \diam{1} Z)$.
\begin{proposition}\label{prop:nonbounds-simple}
$\model,q \satisf \mu Z.(\phi \lor \diam{A} Z)$ does not imply $\model,q \satisf \coop[\ir]{A}\Sometm\phi$.
The converse implication does not hold either.
\end{proposition}

\extended{
  \begin{remark}\label{rem:objective}
  IS THAT REALLY TRUE?
  Interestingly, $tr_1$ provides a correct lower bound for $\coop[\ir]{a}\Sometm\phi$ in the ``objective'' semantics of \atl[\ir] from~\cite{Jamroga03FAMAS,Bulling14comparing-jaamas}, as long as $\diam{a}$ in the translation is interpreted according to the standard ``subjective'' semantics.
  \end{remark}
}


\begin{figure}[t]\centering
\begin{tikzpicture}[transform shape, scale = 0.9]
\input{model-tr1-incorrect}
\end{tikzpicture}
\caption{\ICGS\ $\model_0$: a counterexample for $tr_1$}
\label{fig:tr1}
\end{figure}

\begin{figure}[t]\centering
\begin{tikzpicture}[transform shape, scale = 0.85]
\input{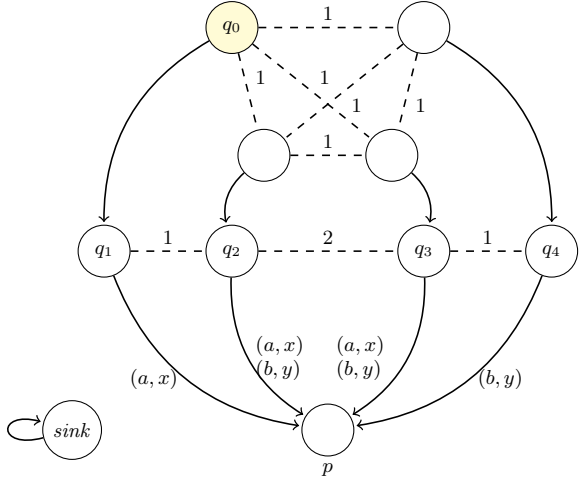}
\end{tikzpicture}
\caption{$\model_1$: a counterexample for $tr_2$}
\label{rysanty}
\end{figure}

Consider now a slightly stronger fixpoint specification:
$\trlow_2(\coop[\ir]{A}\Sometm\phi) = \mu Z.(E_A\phi \lor \diam{A} Z)$.
This new translation works to an extent, as the following proposition shows.
\begin{proposition}\label{prop:reachsimple}
\mbox{}\extended{Let $A\subseteq\Agt$ and $q\in\States$. The following holds:}

\begin{enumerate}
\item\label{it:reachsimp1}
  $\model,q \satisf \mu Z.(E_\emptyset\phi \lor \diam{\emptyset} Z)$\ iff\ $\model,q \satisf \coop[\ir]{\emptyset}\Sometm\phi$;
\item\label{it:reachsimp2}
 If $|A| = 1$, then $\model,q \satisf \mu Z.(E_A\phi \lor \diam{A} Z)$ implies $\model,q \satisf \coop[\ir]{A}\Sometm\phi$, but the converse does not \extended{universally }hold;\footnote{
   Note that, for $A=\set{a}$, $E_A\phi$ is equivalent to $K_a\phi$.}
\item\label{it:reachsimp3}
  If $|A| > 1$, then $\model,q \satisf \mu Z.(E_A\phi \lor \diam{A} Z)$ does not imply $\model,q \satisf \coop[\ir]{A}\Sometm\phi$, and vice versa.
\end{enumerate}
\end{proposition}
\begin{proof}
\textbf{\underline{Case~\ref{it:reachsimp1}}:}
follows from the fact that for the empty coalition the $\ir$--reachability is equivalent to the $IR$--reachabi\-lity,
which in turn has a fixpoint characterization\extended{ in $\AMC$~\cite{Alur02ATL}}.

\smallskip\noindent
\textbf{\underline{Case~\ref{it:reachsimp2}}:}
Let us assume that $A = \{a\}$ for some $a\in\Agt$.
We define the sequence $\{F_j\}_{j\in\Nat}$ of \afAMCi formulae s.t.
$F_0 = \knowmod_a{\phi}$ and $F_{j+1} = F_0 \lor \diam{a} F_j$, for all $j\ge 0$.
From Kleene fixed-point theorem we have
$\denotation{\mu Z.(\knowmod_a \phi \lor \diam{a} Z)} = \bigcup_{j=0}^\infty \denotation{F_j}$,
and $\{\denotation{F_j}\}_{j\in\Nat}$ is a non-decreasing monotone
sequence of subsets of $\States$.
%
%
Now, we prove that for each $j\in\Nat$ there exists a partial strategy $s_a^j$ s.t.
$dom(s_a^{j}) = \denotation{F_{j}}$, $\forall q\in dom(s_a^j) \;\forall\lambda\in \outir(q, s_a^j)\; \exists{k\le j} \;\lambda[k]\satisf{\phi}$,
              and $s_a^j\subseteq s_a^{j+1}$.
The proof is by induction on $j$. We constructively build $s_a^{j+1}$ from $s_a^j$ for each $j\in\Nat$.
The base case is trivial.
For the inductive step, firstly observe that for each $j\in\Nat$ if $q\in\denotation{F_j}$, then
$[q]_{\sim_a}\subseteq\denotation{F_j}$.
As $\sim_a$ is an equivalence relation,
for each $q\in \denotation{F_{j+1}}$ either $[q]_{\sim_a}\subseteq \denotation{F_j}$ or
$[q]_{\sim_a}\subseteq \denotation{F_{j+1}}\setminus \denotation{F_j}$.
In the first case we put $s_a^{j+1}(q) = s_a^j(q)$.
In the second case, we know that there exists a strategy $s_a^{q}$ s.t.
$\forall\lambda\in \outir(q, s_a^{q})\; \lambda[1]\in\denotation{F_j}$.
We thus put $s_a^{j+1}(q') = s_a^{q}(q')$ for all $q'\in[q]_{\sim_a}$, which concludes the inductive proof.

We finally define the partial strategy $s_a = \bigcup_{j\in\Nat}s_a^j$.
For each $q\in\States$ s.t.
$\model,q\satisf\mu Z.(\knowmod_a\phi \lor \diam{a} Z)$,
either $\model,q\satisf{\phi}$, or ${\phi}$ is reached
along each path consistent with any extension of $s_a$ to a full
strategy.

For the converse implication, take model $M$ in~\cite[Figure~1]{Bulling11mu-ijcai}, and observe that $M,q_0 \models \coop[\ir]{1}\Sometm\prop{p}$ but $M,q_0 \not\models \mu Z.(K_1\prop{p} \lor \diam{1} Z)$.

\smallskip\noindent
\textbf{\underline{Case~\ref{it:reachsimp3}}:}
Consider the \ICGS $\model_1$ presented in
Figure~\ref{rysanty}. We assume that $\prot_1(q) = \{a,b\}$ and $\prot_2(q) = \{x,y\}$, for $q\in\{q_1,q_2,q_3,q_4\}$. In the remaining states the protocols
allow only one action.
For clarity, we omit from the figure the transitions leaving the
states $q_1,q_2,q_3$, and $q_4$, leading to state $\sink$.
Assume now $\phi\equiv p$.
Note that
$\model,q_0\satisf\mu Z.(\everyknowmod_{\{1,2\}}{\phi} \lor \diam{\{1,2\}} Z)$
and
$\model,q_0\not\satisf\coop[\ir]{1,2}\Sometm \phi$.
For larger coalitions $A$, we extend the model with a sufficient number of spurious (idle) agents.

For the other direction, use the counterexample from Case~\ref{it:reachsimp2}, extended with appropriately many spurious agents.
\extended{This concludes the proof of the case and the proposition.}
\end{proof}

\begin{figure}[t]
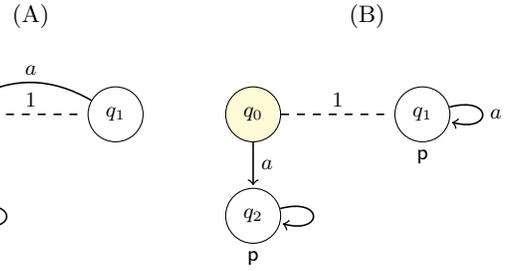
\centering
\begin{tabular}{@{}ccc}
(A) & $\quad$ & (B) \\ \\
\begin{tikzpicture}[transform shape, scale = 0.9]
  \input{model-next-fails.tex}
\end{tikzpicture}
 & &
\begin{tikzpicture}[transform shape, scale = 0.9]
  \input{model-steadfast-fails.tex}
\end{tikzpicture}
\end{tabular}
\caption{Lower bounds are not tight: (A) $\model_2$;\quad (B) $\model_3$}
\label{fig:reaspic}
\end{figure}

As Propositions~\ref{prop:nonbounds-simple} and~\ref{prop:reachsimple} show, translation $\trlow_2$ provides lower bounds for \ATLir verification only in a limited number of instances.
Also, the bound is rather loose, as the following example demonstrates.
\begin{example}\label{ex:reaspic}
Consider the single-agent \ICGS $\model_2$ presented in Figure~\ref{fig:reaspic}A.
The sole available strategy, in which agent $1$ selects always action $a$, enforces eventually reaching $\prop{p}$, i.e.,
$\model_2,q_0 \satisf \coop[\ir]{1}\Sometm\prop{p}$.
On the other hand, $\model_2,q_0 \not\models \mu Z.(K_1\prop{p} \lor \diam{1} Z)$.
This is because the next-step operator in \ATLir requires reaching $p$ simultaneously from \emph{all} the states indistinguishable
from $q_0$, whereas $p$ is reached from $q_0,q_1$ in one and two steps, respectively.
\end{example}

\subsection{Steadfast Next Step Operator}\label{sec:steadfast}


\extended{
  ???
  As we show, the new operator retains the basic properties of the next-step operator of \ATLir
  and the fixed-point reachability
  defined using $\diamStead{\cdot}$ 
  implies $\ir$--reachability.

  \paragraph{Persistent Next-Step Operator}
  Before we present the goal construct, we introduce
  an intermediate next-step operator
  $\diamPers{\cdot}$ whose aim is to alleviate the following limitation of \afAMCi.

  Let $s_A\in\strats_A$ be a strategy for $A\subseteq\Agt$ and
  $Q\subseteq\States$. From
  $\model = \langle\Agt, \States, \Props, \V, \Actions, \prot, \outf, \{\sim_a\mid a\in\Agt\} \rangle$
  we build the restricted model
  $\model_{s_A}^Q = \langle\Agt, \States, \Props', \V', \Actions, \prot', \outf', \{\sim_a\mid a\in\Agt\} \rangle$
  as follows.
  Firstly, in each state $q\in\States$ and agent $a\in A$ we limit the set of actions allowed
  to $a$ to the one selected by $s_a$, i.e., $\prot_a(q)' = \{s_a(q)\}$.
  Secondly, we introduce a new, fresh proposition $\qprop$ with labeling $\V'(\qprop) = Q$.
  All the remaining components of $\model_{s_A}^Q$ are inherited from $\model$.

  Let $q\in\States$, $\xval\in\xvarsvals$, and $\phi\in\afAMCi$.
}

To obtain a tighter lower bound, and one that works universally, we introduce a new modality. $\diamStead{A}$ can be seen as a semantic variant of the next-step ability operator $\diam{A}$ where:\ (i) agents in $A$ look for a short-term strategy that succeeds from the ``common knowledge'' neighborhood of the initial state (rather than in the ``everybody knows'' neighborhood), and\ (ii) they are allowed to ``steadfastly'' pursue their goal in a variable number of steps within the indistinguishability class.
In this section, we propose the semantics of $\diamStead{A}$ and show how to revise the lower bound. Some additional insights are provided in Section~\ref{sec:bounds-disc}.

We begin by defining the auxiliary function $\Reach$ so that $q\in\Reach_M(s_A,Q,\phi)$ collects all $q\in Q$ such that all the paths executing $s_A$
from $q$ eventually reach $\phi$ without leaving $Q$, except possibly for the last step:

\vspace{1mm}
\noindent
$\hspace{-0.1cm}\Reach_M(s_A,Q,\phi) =
  \{ q\in Q \mid \forall\lambda\in out(q,s_A) \\
  \mbox{}\qquad\qquad \exists i\ .\ M,\lambda[i] \satisf \phi\text{ and } \forall 0\le j<i\ .\ \lambda[j]\in Q\}$.

\smallskip\noindent
\extended{
  Our first attempt at the new operator is as follows:
  %
  \begin{itemize}
  \item $\model,q\satisf\diamPers{A}\phi$ iff there exists $s_A\in\strats_A^\ir$ such that $\Reach_M(s_A,[q]_{\sim_A^E},\phi) = [q]_{\sim_A^E}$.
  \end{itemize}
  Note that
  $\model,q\satisf\diamPers{A}\phi$ holds iff there exists a strategy whose every outcome path starting in $[q]_{\sim_E^A}$ eventually reaches $\phi$ without leaving $[q]_{\sim_E^A}$.

  \begin{proposition}
  If $\model,q \models \diamPers{A}\phi$, then $\model,q \satisf \coop[\ir]{A}\Sometm\phi$.
  \end{proposition}

  %
  Thus, $\diamPers{A}\phi$ actually provides a lower bound for $\coop[\ir]{A}\Sometm\phi$,
  though the bound is very weak in the sense that $\diamPers{A}\phi$ will seldom hold.
  Unfortunately, if we use it to replace $\diam{A}$ in the ``natural'' fixpoint unfolding of $\coop[\ir]{A}\Sometm\phi$
  from $\trlow_2$, then it provides neither lower nor upper bound\extended{ for verification}:
  \begin{proposition}
  $\model,q \satisf \mu Z.(E_A\phi \lor \diamPers{A} Z)$ does not imply $\model,q \satisf \coop[\ir]{A}\Sometm\phi$, and vice versa.
  \end{proposition}
  \begin{proof}
  For the \ICGS in Figure~\ref{rysanty} we have $\model,q\satisf \mu Z.(E_A \phi \lor \diamPers{A} Z)$
  and $\model,q\not\satisf \coop[\ir]{A}\Sometm \phi$.
  Recall that the method of unification of strategies for the case of single agent,
  presented in the proof of Proposition~\ref{prop:reachsimple}, was based on the fact that the relation
  of everybody knows $\sim_A$ is an equivalence relation if $|A| = 1$.
  The other direction follows analogously to cases~\ref{it:reachsimp2} and~\ref{it:reachsimp3} of Proposition~\ref{prop:reachsimple}.
  \end{proof}

  However, a small modification of the semantics suffices to obtain what we need.
} 
The \emph{steadfast next-step operator} $\diamStead{A}$ is defined\extended{ by replacing the ``everybody knows'' neighborhood of $q$ with its common knowledge neighborhood,} as follows:
\begin{itemize}
\item $\model,q\satisf\diamStead{A}\phi$ iff there exists $s_A\in\strats_A^\ir$ such that $\Reach_M(s_A,[q]_{\sim_A^C},\phi) = [q]_{\sim_A^C}$.
\end{itemize}
%
%
Now we can propose our ultimate attempt at the lower bound for reachability goals:
$\trlow_3(\coop[\ir]{A}\Sometm\phi) = \mu Z . (E_A\phi \lor \diamStead{A} Z)$, with the following result.

\begin{proposition}\label{prop:reachthm}
If $\model,q\models \mu Z . (E_A\phi \lor \diamStead{A} Z)$, then $\model,q\satisf \coop[\ir]{A}\Sometm\phi$.
The converse does not universally hold.
\end{proposition}
\begin{proof}
The proof is similar to the proof of Proposition~\ref{prop:reachsimple}.
As previously, we define a sequence $\{F_j\}_{j\in\Nat}$ of \afAMCi formulae s.t. $F_0 = \everyknowmod_A{\phi}$
and $F_{j+1} = F_0 \lor \diamStead{A} F_j$, for all $j\ge 0$.
We also use a sequence $\{H_j\}_{j\in\Nat}$ with $H_j = \diamStead{A} F_j$.
From Kleene fixed-point theorem we have
$\denotation{\mu Z . (E_A\phi \lor \diamStead{A} Z)}
\linebreak = \bigcup_{j=0}^\infty \denotation{F_j}
= \denotation{F_0} \cup \bigcup_{j=0}^\infty \denotation{H_j}$.
Observe that, as $\sim_C^A$ is an equivalence relation, we have for each $q\in\States$
and $j\in\Nat$ that
if $[q]_{\sim_C^A}\cap \denotation{H_{j}}\ne \emptyset$, then
$[q]_{\sim_C^A}\subseteq \denotation{H_{j}}$.

We prove that for each $j\in\Nat$ there exists a partial strategy $s_A^j$ s.t.
$dom(s_A^{j}) = \denotation{H_{j}}$,
$\forall q\in dom(s_A^j) \;\forall\lambda\in \outir(q, s_A^j)\; \exists k\in\Nat\;\lambda[k]\satisf\everyknowmod_A{\phi}$,
and $s_A^j\subseteq s_A^{j+1}$.
The proof is by induction on $j$. In the base case of $H_0 = \diamStead{A} \everyknowmod_A{\phi}$
observe that if $q\in \denotation{H_0}$ then there exists
a partial strategy $s_A^{0,q}$ with $dom(s_A^{0,q}) = [q]_{\sim_C^A}$ s.t. every
$\lambda\in \outir(q, s_A^{0,q})$ stays in $[q]_{\sim_C^A}$ until it reaches a state where $\everyknowmod_A{\phi}$ holds.
We can now
define $s_A^{0} = \bigcup_{[q]_{\sim_C^A} \in \States/\sim_C^A} s_A^{0,q}$
which is uniform, and reaches $E_A\phi$ on all execution paths.
%
For the inductive step, we divide the construction of $s_A^{j+1}$ in two cases.
Firstly, if $q\in\denotation{H_j}$, then we put $s_A^{j+1}(q) = s_A^{j}(q)$.
Secondly, let $q\in\denotation{H_{j+1}}\setminus\denotation{H_j}$.
In this case there exists a partial strategy $s_A^{j+1,q}$ with
$dom(s_A^{j+1,q}) = [q]_{\sim_C^A}$ s.t. each outcome $\lambda\in \outir(q, s_A^{j+1,q})$
stays in $[q]_{\sim_C^A}$ until it reaches a state $q'$ s.t. either $q'\satisf\everyknowmod_A{\phi}$
or $q'\in \denotation{H_j}$.
In the latter, from the inductive assumption we know that following $s_A^{j+1}$
always leads to reaching $\everyknowmod_A{\phi}$
without leaving $\denotation{H_j}$.
We thus take
$s_A^{j+1} = \bigcup_{[q]_{\sim_C^A} \in \States/\sim_C^A} s_A^{j+1,q}$
which, again, is uniform, and reaches $E_A\phi$ on all execution paths. This concludes the inductive part of the proof.

Finally, we build a partial strategy $s_A = \bigcup_{j\in\Nat} s_A^j$,
whose any extension 
is s.t. for each $q\in\States$, if
$\model,q\satisf\mu Z . (E_A\phi \lor \diamStead{A} Z)$,
then a state in which $\everyknowmod_A{\phi}$ holds is eventually reached
along each outcome path $\lambda\in \outir(q, s_A')$.
This concludes the proof of the implication.

To see that the converse does not hold, consider model $\model_3$ in Figure~\ref{fig:reaspic}B. We have that $\model_3,q_0\satisf \coop[\ir]{1}\Sometm\prop{p}$, but $\model_3,q_0\not\models \mu Z . (K_1\prop{p} \lor \diamStead{1} Z)$.
\end{proof}

Thus, $\trlow_3$ indeed provides a universal lower bound for reachability goals expressed in \ATLir.

\subsection{Lower Bounds for ``Always'' and ``Until''}

%
%

So far, we have concentrated on\extended{ the intuitive case of} reachability goals\extended{, expressed with the strategic operator $\coop[\ir]{A}\Sometm$}. We now extend the main result to all the modalities of \ATLir:
\begin{theorem}\label{prop:lowerbounds}
\mbox{}

\begin{enumerate}
\item\label{it:approxalways} If $\model,q \satisf \nu Z.(\commonknowmod_A\phi \land \diamStead{A} Z)$,
  then $\model,q \satisf \coop[\ir]{A}\Always\phi$;
\item\label{it:approxuntil} If $\model,q \satisf \mu Z.\big(\everyknowmod_A\phi \lor(\commonknowmod_A\psi \land \diamStead{A} Z)\big)$, then $\model,q \satisf \coop[\ir]{A}\psi\Until\phi$.
\end{enumerate}
\end{theorem}
\begin{proof}
\textbf{\underline{Case~\ref{it:approxalways}}:}
Let us define the sequence $\{G_j\}_{j\in\Nat}$ of formulae s.t.
$G_0 = \commonknowmod_A\phi$ and $G_{j+1} = G_0 \land \diamStead{A} G_j$,
for all $j\ge 0$.
From Kleene fixed-point theorem,
$\denotation{\nu Z.(\commonknowmod_A\phi \land \diamStead{A} Z)} = \bigcap_{j=0}^\infty \denotation{G_j}$.
It suffices to prove that for each $j\in\Nat$ there exists a strategy $s_A^j$ s.t.
$\forall q\in \denotation{G_j} \;\forall\lambda\in \outir(q, s_A^j)\; \forall{0\le k\le j}$ $\lambda[k]\satisf\phi$.
The proof is by induction on $j$, with the trivial base case.
Assume that the inductive assumption holds for some $j\in\Nat$.
From the definition of the steadfast next-step operator
we can define for each equivalence class $[q]_{\sim_C^A}\in \denotation{G_{j+1}}/\sim_C^A$
a partial strategy $s_{A}^{q, j+1}$ s.t.
$\forall q'\in [q]_{\sim_C^A}\;\forall\lambda\in \outir(q, s_{A}^{q, j+1}) \;\lambda[1]\in\denotation{G_j}$.
We now construct \\
\centerline{$s_{A}^{j+1} = \bigcup_{[q]_{\sim_C^A}\in \denotation{G_{j+1}}/\sim_C^A} s_{A}^{q, j+1}
\cup
{s_{A}^j}|_{\denotation{\commonknowmod_A\phi}\setminus\denotation{G_j}}$.}\\
Intuitively, $s_{A}^j$ enforces that a path leaving each $q\in\denotation{G_{j+1}}$ stays within $\denotation{\commonknowmod_A\phi}$ for at least $j$ steps.
Moreover, $s_{A}^j\subseteq s_{A}^{j+1}$ for all $j$.
Thus, $s_A = \bigcup_{j\in\Nat} s_A^{j}$ enforces that a path leaving each $q\in\bigcup_{j\in\Nat}\denotation{G_{j}}$ stays within $\denotation{\commonknowmod_A\phi}$ for infinitely many steps, which concludes the proof.
\extended{That's a stronger property than we need! Can we somehow weaken the construction?}
%
Note that the correctness of the construction relies the fact that $\sim_C^A$ is
an equivalence relation.

\textbf{\underline{Case~\ref{it:approxuntil}}:}
analogous to\extended{ that of} Proposition~\ref{prop:reachthm}.
\end{proof}

\extended{
  In fact, a closer inspection of the above proof shows that a stronger result can be obtained.
  Namely, if we define a new logic by extending \ATLir with the operator of
  common knowledge (recall~\cite{Bulling11mu-ijcai} that the everybody knows operator
  is a derived modality in \ATLir), then:

  \begin{itemize}
  \item $\model,q\satisf\coop[ira]{A} \Always\phi$
  implies
  $\model,q\satisf\coop[\ir]{A} \Always\commonknowmod_A\phi$,
  \item $\model,q\satisf\coop[ira]{A} \psi\Until\phi$
  implies
  $\model,q\satisf\coop[\ir]{A} \everyknowmod_A\psi\Until\commonknowmod_A\phi$.
  \end{itemize}
}

\section{Discussion \& Properties}\label{sec:bounds-disc}

Theorem~\ref{prop:lowerbounds} shows that $\trlow_3(\phi)$ provides a correct lower bound of the value of $\phi$ for all formulae of \ATLir. In this section, we discuss the tightness of the approximation from the theoretical point of view. An empirical evaluation will be presented in Section~\ref{sec:experiments}.

\subsection{Comparing $\trlow_2$ and $\trlow_3$ for Reachability Goals}\label{sec:bounds-compare}

Translation $\trlow_3$ updates $\trlow_2$ by replacing the standard next-step ability operator $\diam{A}$ with the ``steadfast next-step ability'' $\diamStead{A}$.
The difference between\extended{ the semantics of} $\diam{A}\phi$ and $\diamStead{A}\phi$ is twofold.
First, $\diam{A}\phi$ looks for a winning\extended{ short-term} strategy in the ``everybody knows'' neighborhood of a given state (i.e., $[q]_{\sim_A^E}$), whereas $\diamStead{A}\phi$ looks at the ``common knowledge'' neighborhood (i.e., $[q]_{\sim_A^C}$).
Secondly, $\diamStead{A}$ allows to ``zig-zag'' across $[q]_{\sim_A^C}$ until a state satisfying $\phi$ is found.

Actually, the first change would suffice to provide a universally {correct} lower bound for \ATLir.
The second update makes it \emph{more useful} in models where agents may not see the occurrence of some action, such as $\model_2$ of Figure~\ref{fig:reaspic}A.
To see this formally, we show that $\trlow_3$ provides a strictly tighter approximation than $\trlow_2$ on singleton coalitions:
\begin{proposition}\label{prop:tighter}
For $A=\set{a}$, if $\model,q\models \mu Z . (K_a\phi \lor \diam{a} Z)$, then $\model,q\models \mu Z . (K_a\phi \lor \diamStead{a} Z)$.
The converse does not universally hold.
\end{proposition}
\begin{proof}
It suffices to observe that $M,q\satisf \diam{a}\phi$ implies $M,q\satisf \diamStead{a}\phi$,
for any $\phi\in\afAMCi$. Note that this is true only for single-agent
coalitions.
For the converse, notice that in \ICGS $\model_2$ from
Figure~\ref{fig:reaspic}A we have $\model_2,q_0\satisf \mu Z . (K_1\prop{p} \lor \diamStead{1} Z)$ and
$\model_2,q_0\not\satisf \mu Z . (K_1\prop{p} \lor \diam{1} Z)$.
\end{proof}


On the other hand, if agent $a$ always sees whenever an action occurs, then $\trlow_2$ and $\trlow_3$ coincide for $a$'s abilities. Formally, let us call \ICGS $M$ \emph{lockstep for $a$} if, whenever there is a transition from $q$ to $q'$ in $M$, we have $q\not\sim_a q'$.
The following is straightforward.
\begin{proposition}\label{prop:lockstep}
If $M$ is lockstep for $a$, then $\model,q\models\diam{a}\phi$ iff $\model,q\models\diamStead{a}\phi$.
In consequence, $\model,q\models \trlow_2(\coop{a}\Sometm\phi)$ iff $\model,q\models\trlow_3(\coop{a}\Sometm\phi)$.
\end{proposition}

\subsection{When is the Lower Bound Tight?}

An interesting question is: what is the subclass of \ICGS's for which $\trlow_3$ is tight, i.e., the answer given by the approximation is exact?
We address the question only partially here. In fact, we characterize a subclass of \ICGS's for which $\trlow_3$ is certainly \emph{not} tight, by the necessary condition below.

Let $\gamma\equiv\Always\psi$ or $\gamma\equiv\psi_1\Until\psi_2$ for some $\psi,\psi_1,\psi_2\in\ATLir$.
We say that strategy $s_A\in\strats_A^\ir$ is \emph{winning for $\gamma$ from $q$} if it obtains $\gamma$ for all paths in $out^{\ir}(q,s_A)$.
Moreover, for such $s_A$, let $RR(q,s_A,\gamma)$ be the\extended{ set of} \emph{relevant reachable states of $s_A$ in the context of $\gamma$}, defined as follows:\
$RR(q,s_A,\Always\psi)$ is the set of states that occur anywhere in $out^{\ir}(q,s_A)$;\
$RR(q,s_A,\psi_1\Until\psi_2)$ is the set of states that occur anywhere in $out^{\ir}(q,s_A)$ before the first occurrence of $\psi_2$.
\begin{proposition}\label{prop:tight}
Let $M$ be a \ICGS, $q\in\States_M$, and $\phi \equiv \coop[\ir]{A}\gamma$.
Furthermore, suppose that $\phi$ and $\trlow_3(\phi)$ are either both true or both false in $M,q$. Then:
\begin{enumerate2}
\item\label{it:nostrat} either no strategy $s_A\in\strats_A^\ir$ is winning for $\gamma$ from $q$,\ or
\item\label{it:recomputable} there is a strategy $s_A\in\strats_A^\ir$ which is winning for $\gamma$ from every $q'\in RR(q,s_A,\gamma)$.
\end{enumerate2}
\end{proposition}
\wj{proof?}

Conversely, the approximation is \emph{not} tight if there are winning strategies, but each of them reaches a intermediate state $q'$ from which no winning substrategy can be computed. This can only happen if some states in\extended{ the epistemic neighborhood} $[q']_{\sim_A^E}$ are not reachable by $s_A$. In consequence, the agents in $A$ forget relevant information that comes alone from the fact that they are executing $s_A$.
%
We will use Proposition~\ref{prop:tight} in Section~\ref{sec:experiments} to show that the few benchmarks existing in the literature are not amenable to our approximations.
\extended{The complete characterization of applicability for our approximation schemeis left for future work.}

\section{Approximation Semantics for \ATLir}

Note that $\model,q \satisf \coop[\ir]{A}\gamma$ always implies $\model,q \satisf E_A\coop[\Ir]{A}\gamma$\extended{ (the proof is straightforward from the semantics)}.
Based on this, and the lower bounds established in Theorem~\ref{prop:lowerbounds}, we propose the \emph{lower approximation} $\trlow$ and the \emph{upper approximation} $\trup$ for \ATLir as follows:

\smallskip
{
\hspace{-0.4cm}
\begin{tabular}{l}
$\trlow(p) = p$,\quad
$\trlow(\neg\phi) = \neg\trup(\phi)$,\quad
$\trlow(\phi\land\psi) = \trlow(\phi) \land \trlow(\psi)$,
\\
$\trlow(\diam{A}\phi) = \diam{A}\trlow(\phi)$,
\\
$\trlow(\coop[\ir]{A}\Always\phi) = \nu Z.(\commonknowmod_A{\trlow}(\phi) \land \diamStead{A} Z)$,
\\
$\trlow(\coop[\ir]{A}\psi\Until\phi) = \mu Z.\big(\everyknowmod_A\trlow(\phi) \lor(\commonknowmod_A\trlow(\psi) \land \diamStead{A} Z)\big)$.
\\
\end{tabular}
}

\medskip
{
\hspace{-0.4cm}
\begin{tabular}{l}
$\trup(p) = p$,\qquad
$\trup(\neg\phi) = \neg\trlow(\phi)$,
\\
$\trup(\phi\land\psi) = \trup(\phi) \land \trup({\psi})$,
\\
$\trup(\diam{A}\phi) = E_A\coop[\Ir]{A}\Next\trup(\phi)$,
\\
$\trup(\coop[\ir]{A}\Always\phi) = E_A\coop[\Ir]{A}\Always\trup(\phi)$,
\\
$\trup(\coop[\ir]{A}\psi\Until\phi) = {E_A\coop[\Ir]{A}}\trup(\psi)\Until\trup(\phi)$.
\\
\end{tabular}
}

\extended{
  \begin{tabular}{l@{\quad}l}
  $\trlow(p) = p$
  &
  $\trup(p) = p$
  \\
  $\trlow(\neg\phi) = \neg\trup(\phi)$
  &
  $\trup(\neg\phi) = \neg\trlow(\phi)$
  \\
  $\trlow(\phi\land\psi) = \trlow(\phi) \land \trlow(\phi)$
  &
  $\trup(\phi\land\psi) = \trup(\phi) \land \trup(\phi)$
  \\
  $\trlow(\diam{A}\phi) = \diam{A}\trlow(\phi)$
  &
  $\trup(\diam{A}\phi) = E_A\coop[\Ir]{A}\Next\trup(\phi)$
  \\
  $\trlow(\coop[\ir]{A}\Always\phi) = \nu Z.(\commonknowmod_A\phi \land \diamStead{A} Z)$
  &
  $\trup(\coop[\ir]{A}\Always\phi) = E_A\coop[\Ir]{A}\Always\trup(\phi)$
  \\
  $\trlow(\coop[\ir]{A}\psi\Until\phi) = \mu Z.\big(\everyknowmod_A\trlow(\phi) \lor(\commonknowmod_A\trlow(\psi) \land \diamStead{A} Z)\big)$
  &
  $\trup(\coop[\ir]{A}\psi\Until\phi) = E_A\coop[\Ir]{A}\coop[\ir]{A}\trup(\psi)\Until\trup(\phi)$.
  \\
  \end{tabular}
} 

\medskip
The following important results can be proved by straightforward induction on the structure of $\phi$.

\begin{theorem}\label{prop:approximation}
For any\extended{ \ICGS $M$, a state in it, and an} \ATLir formula $\phi$:\\
\centerline{$\model,q \models \trlow(\phi)\ \ \Rightarrow\ \ \model,q \models \phi\ \ \Rightarrow\ \ \model,q \models \trup(\phi)$.}
\end{theorem}
\extended{
  \begin{proof}
  The proof follows via a straightforward induction on the structure
  of $\phi$, employing Theorem~\ref{prop:lowerbounds} and
  the observation that $q \satisf \coop[\ir]{A}\gamma$ implies $q \satisf E_A\coop[\Ir]{A}\gamma$
  for all $\gamma\in\{\Next\phi, \Always\phi, \psi\Until\phi\}$.
  \end{proof}
} 


\begin{theorem}\label{prop:approximation-complexity}
If $\phi$ includes only coalitions of size at most 1, then model checking $\trlow(\phi)$ and $\trup(\phi)$ can be done in time $O(|\model|\cdot|\phi|)$.
In the general case, the problem is between \NP and \Deltwo wrt $\max_{A\in\phi}(|\!\sim_A^C\!|)$ and $|\phi|$.
\end{theorem}

Thus, our approximations potentially offer computational advantage\extended{ over direct \ATLir model checking in two cases:} when we consider coalitions whose members have similar knowledge, and especially when verifying abilities of individual agents.

\para{Approximation of abilities under perfect recall.}
In this paper, we focus on approximating abilities based on memoryless strategies.
Approximations might be equally useful for \ATLiR (i.e., the variant of \ATL using uniform perfect recall strategies); we simply begin with the problem that is easier in its exact form.
The high intractability of \ATLiR model checking suggests that a substantial extension will be needed to come up with satisfactory approximations.

\extended{
  Note, however, that if $A$ have a successful memoryless strategy, then they also have a winning perfect recall strategy.
  In consequence, as long as our lower bounds are correct for \ATLir, they are also {correct} for  variant of \ATL (i.e., \ATLiR).
  However, they do not find a perfect recall strategy if no winning memoryless strategy is available.

} 
We also observe that the benchmark in Section~\ref{sec:experiments-voting} is a \emph{model of perfect recall}, i.e., the states\extended{ of the model} explicitly encode the agents' memory of their past observations. In con\-sequence, the memoryless and perfect recall semantics of \ATL coincide. The experimental results suggest that, for such models, verification of perfect recall abilities can be much improved by using the approximations proposed here.

\newcommand{\phivotingthree}{\phi_1}
\newcommand{\phivotingfour}{\phi_2}

\extended{
  \begin{figure*}\centering
  \hspace{-0.75cm}
  \begin{tikzpicture}[>=latex,scale=1.15]
   \input{voting-new.tex}
  \end{tikzpicture}
  \caption{A voter module for experiments}
  \label{fig:votingmodel-new}
  \end{figure*}
} 

\begin{figure}[h]
\hspace{-0.4cm}\resizebox{1.1\columnwidth}{!}{%
\begin{tabular}{|c|HHHc|Hc|Hc|Hc|Hc|Hc|c|c|}
\hline
\multirow{2}{*}{$k$} & \multirow{2}{*}{btime}    & \multirow{2}{*}{\#bsat} & \multirow{2}{*}{rtime}    & \multirow{2}{*}{\#states}  & \multirow{2}{*}{eptime}   & \multirow{2}{*}{tgen}
  & \multicolumn{4}{c|}{Lower approx.\extended{ (\AEMC)}}           & \multicolumn{4}{c|}{Upper approx.\extended{ (\ATLIr)}}      &  \multirow{2}{*}{Match}   &    Exact   \\
\cline{8-15}
\cline{17-17}
  &  &  &  &  &  &  & \#iter & tverif     & \#sat   & result & \#iter & tverif   & \#sat   & result &  &  tg+tv \\
\hline
\hline
$1$           &  &  &  & {15}     &  & 0.001 &  & 0.0001 &  & True   &  & 0.00007 &  & True  & {100\%}   & 0.006         \\ \hline
$2$           &  &  &  & {225}    &  & 0.02  &  & 0.002  &  & True   &  & 0.001   &  & True  & {100\%}   & 14.79         \\ \hline
$3$           &  &  &  & {3375}   &  & 0.50  &  & 0.14   &  & True   &  & 0.03    &  & True  & {100\%}   & \tmout        \\ \hline
$4$           &  &  &  & {50625}  &  & 14.39 &  & 22.78  &  & True   &  & 0.77    &  & True  & {100\%}   & \tmout        \\ \hline
\end{tabular}%
}

\caption{Experimental results for simple voting model ($\phivotingthree$)}
\label{fig:resultsvoting-3}
\end{figure}

\begin{figure}[h]
\hspace{-0.4cm}\resizebox{1.1\columnwidth}{!}{%
\begin{tabular}{|c|HHHc|Hc|Hc|Hc|Hc|Hc|c|c|}
\hline
\multirow{2}{*}{$k$} & \multirow{2}{*}{btime}    & \multirow{2}{*}{\#bsat} & \multirow{2}{*}{rtime}    & \multirow{2}{*}{\#states}  & \multirow{2}{*}{eptime}   & \multirow{2}{*}{tgen}
  & \multicolumn{4}{c|}{Lower approx.\extended{ (\AEMC)}}           & \multicolumn{4}{c|}{Upper approx.\extended{ (\ATLIr)}}      &  \multirow{2}{*}{Match}   &    Exact   \\
\cline{8-15}
\cline{17-17}
  &  &  &  &  &  &  & \#iter & tverif     & \#sat   & result & \#iter & tverif   & \#sat   & result &  &  tg+tv \\
\hline
\hline
$1$           &  &  &  & {15}     &  & 0.001 &  & 0.00005 &  & False   &  & 0.00003 &  & False  & {100\%}   & 0.005        \\ \hline
$2$           &  &  &  & {225}    &  & 0.02  &  & 0.0005  &  & False   &  & 0.0003  &  & False  & {100\%}   & 0.02         \\ \hline
$3$           &  &  &  & {3375}   &  & 0.50  &  & 0.01    &  & False   &  & 0.007   &  & False  & {100\%}   & 0.04         \\ \hline
$4$           &  &  &  & {50625}  &  & 14.39 &  & 0.94    &  & False   &  & 0.12    &  & False  & {100\%}   & 0.12         \\ \hline
\end{tabular}%
}

\caption{Experimental results for simple voting model ($\phivotingfour$)}
\label{fig:resultsvoting-4}
\end{figure}

\section{Experimental Evaluation}\label{sec:experiments}

Theorem~\ref{prop:approximation} and Proposition~\ref{prop:approximation-complexity} validate the approximation semantics theoretically. In this section, we back up the theoretical results by looking at how well the approximations work in practice.
We address two issues: the \emph{performance} and the \emph{accuracy} of the approximations.

\subsection{Existing Benchmarks}\label{sec:existing}

%
The only publicly available tool that provides verification of \ATL with imperfect information is MCMAS~\cite{Lomuscio06uniform,Raimondi06phd,Lomuscio09mcmas,Lomuscio15mcmas}. We note, however, that imperfect information strategies are not really at the heart of the model-checker, the focus being on verification of \CTLK and \ATLK with perfect information strategies.
More dedicated attempts produced so far only experimental algorithms, with preliminary performance results reported in~\cite{Pilecki14synthesis,Busard14improving,Huang14symbolic-epist,Busard15reasoning,Pilecki17smc}.
Because of that, there are few benchmarks for model checking \ATLir, and few experiments have actually been conducted.

The classes of models typically used to estimate the performance of \ATLir model checking are \textbf{TianJi}~\cite{Raimondi06phd,Busard14improving} and \textbf{Castles}~\cite{Pilecki14synthesis}.
The properties to be verified are usually reachability properties, saying that Tian Ji can achieve a win over the king (in \textbf{TianJi}), or that a given coalition of workers can defeat another castle (for \textbf{Castles}).
We observe that both \textbf{TianJi} and \textbf{Castles} \emph{do not satisfy} the necessary condition in Proposition~\ref{prop:tight}.
This is because the states of the model do not encode some relevant information about the actions that have been already played by the coalition.
Thus, even one step before winning the game, the players take into account also some (possibly losing) states that couldn't be reached by the strategy that they are executing.

This means that the \AEMC approximations, proposed in this paper, are not useful for \textbf{TianJi} and \textbf{Castles}.
It also means that the benchmarks arguably do not capture realistic scenarios.
We usually do not want to assume agents to forget \emph{their own actions} from a few steps back.
In the remainder, we propose several new benchmarks that can be used to evaluate our approximation scheme.

Finally, we note that most experiments reported in the literature use very simple input formulae (no nested strategic modalities; singleton coalitions or groups of agents with identical indistinguishability relations).
As the results show, verification of such formulae is complex enough -- see the performance of exact model checking in the rest of this section.

\extended{
  our lower bound translation $\trlow(\phi)$ returns $\mathit{false}$ in most models of those two classes, regardless of the truth value of the original specification $\phi$. This is because formulae of \AEMC specify abilities that are \emph{recomputable} in the sense that the players can, at any stage, compute the remaining part of the winning strategy from the available observations.

  Note that in many cases we do \emph{not} want to reason about agents who forget relevant information. This is especially true in security applications, cf.~our simple voting example. When verifying existence of coercion strategies, we should not assume the attacker to be dumb and forgetful. In the remainder, we propose and employ a novel benchmark that shares much structural characteristics of the voting scenario.
} 

\subsection{Verifying the Simple Voting Scenario}\label{sec:experiments-voting}


For the first benchmark, we adapt the simple voting scenario from Example~\ref{ex:voting-model}.
The model consists of $k+1$ agents ($k$ voters $v_1,\dots,v_k$, and 1 coercer $c$).
The module of voter $v_i$ implements the transition structure from Figure~\ref{fig:votingmodel}, with three modifications. First, the voter can at any state execute the ``idle'' action \textit{wait} (this is needed to ensure uniformity of the resulting \ICGS).
In consequence, synchronous voting as well as interleaving of votes is allowed.
Secondly, in states $q_3,\dots,q_6$, the coercer's action $np$ (``no punishment'') leads to an additional final state ($q_7',\dots,q_{10}'$), labeled accordingly. Thirdly, the old and new leaves in the structure (i.e., $q_7,\dots,q_{10},q_7',\dots,q_{10}'$) are labeled with an additional atomic proposition \prop{finish_i}.

As specifications, we want to use the properties saying that:\
(i) the coercer can force the voter to vote for candidate 1 or else the voter is punished, and\
(ii) the voter can avoid voting for candidate 1 and being punished (cf.~Example~\ref{ex:voting-formulae}).
Note, however, that the model used for the experiments is an unconstrained product of the voter modules.
Thus, it includes also paths that were absent in the CEGS $\Mvote$ from Example~\ref{ex:voting-model} (in particular, ones where a voter executes \textit{wait} all the time).
To deal with this, we modify the specifications from Example~\ref{ex:voting-formulae} so that they discard such paths:
\begin{enumerate2}
\item $\phivotingthree \equiv \coop[\ir]{c}\Always \big((\prop{finish_i}\land\neg\prop{pun_i}) \then \prop{vote_{i,1}}\big)$ which always holds in the voting scenario,
\item $\phivotingfour \equiv \coop[\ir]{v_i}\Sometm \big(\prop{finish_i} \land \neg\prop{pun_i} \land \neg\prop{vote_{i,1}}\big)$ which is always false.
\end{enumerate2}


The results of experiments for\extended{ formula} $\phivotingthree$ are shown in Figure~\ref{fig:resultsvoting-3}, and for $\phivotingfour$ in Figure~\ref{fig:resultsvoting-4}.
The columns present the following information: parameter of the model (the number of voters $k$), size of the state space (\#states), generation time for models (tgen), time and output of verification (tver, result) for model checking the lower approximation $\trlow(\phi)$, and similarly for the upper approximation $\trup(\phi)$; the percentage of cases where the bounds have matched (match), and the total running time of the exact \ATLir model checking for $\phi$ (tg+tv). The running times are given in seconds. \emph{Timeout} indicates that the process did not terminate in 48 hours (!).

The computation of the lower and upper approximations was done with a straightforward implementation (in Python 3) of the fixpoint model checking algorithm for \AEMC and \ATLIr, respectively.
We used the explicit representation of models, and the algorithms were not optimized in any way.
The exact \ATLir model checking was done with MCMAS 1.2.2 in such a way that the underlying CEGS of the ISPL code was isomorphic to the explicit models used to compute approximations.
The subjective semantics of \ATLir was obtained by using the option \emph{-atlk 2} and setting the initial states as the starting indistinguishability class for the proponent.
All the tests were conducted on a PC with an Intel Core i5-2500 CPU with dynamic clock speed of 3.30 GHz up to 3,60 GHz, 8 GB of RAM (two modules DDR3, 1600 MHz bus clock), and Windows 10 (64bit).

\para{Discussion of results.}
Exact model checking with MCMAS performed well on the inputs where no winning strategy existed (formula $\phivotingfour$), but was very bad at finding the existing winning strategy for formula $\phivotingthree$.
In that case, our approximations offered huge speedup.
Moreover, the approximations actually found the winning strategy in all the tested instances, thus producing fully conclusive output.
This might be partly due to the fact that the scenario uses \emph{perfect recall models}, i.e., ones encoding perfect memory of players explicitly in their local states.

\begin{figure}[h]\centering
\includegraphics[width=6cm]{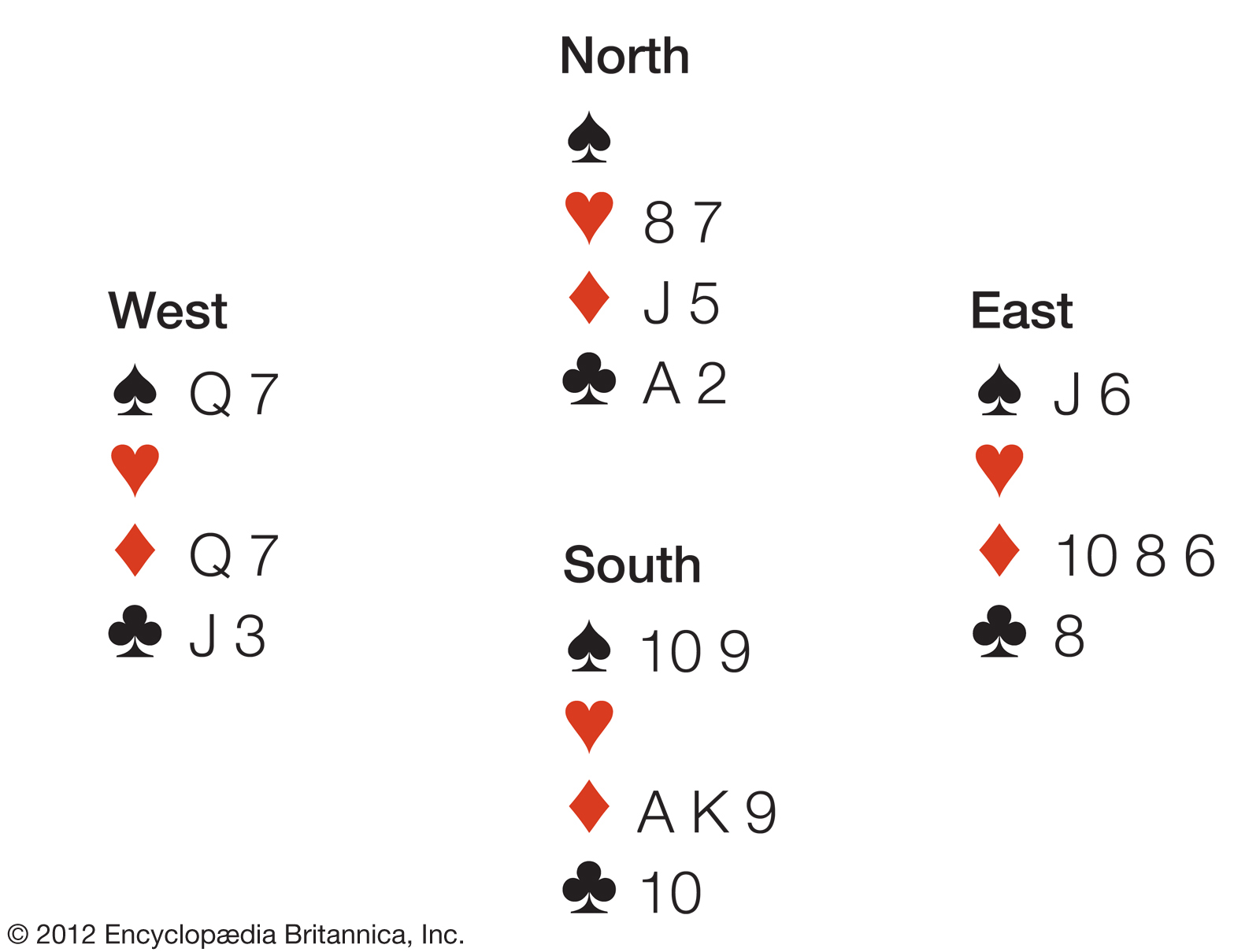}
\caption{Example 6-endplay in bridge: }
\label{fig:bridge}
\end{figure}

\subsection{Bridge Endplay}\label{sec:bridge}


We use bridge play scenarios of a type often considered in bridge handbooks and magazines. The task is to find a winning strategy for the declarer, usually depicted at the South position (\South), in the $k$-endplay of the game, see Figure~\ref{fig:bridge} for an example. The deck consists of $4n$ cards in total ($n$ in each suit),\footnote{
  In real bridge, $n=13$\extended{, but we keep it variable to study how the model checking algorithms will scale up for different sizes of the deck}. }
and the initial state captures each player holding $k$ cards in their hand, after having played $n-k$ cards. This way we obtain a family of models, parameterized by the possible values of $(n,k)$.
A NoTrump contract is being played; the declarer wins if she takes more than $k/2$ tricks in the endplay.

The players' cards are played sequentially (clockwise). \South plays first at the beginning of the game. Each next trick (i.e., the set of four played cards, one per player) is opened by the player who won the latest trick.
The declarer handles her own cards and the ones of the dummy (\North). The opponents (\West and \East) handle their own hands each. The cards of the dummy are visible to everybody; the other hands are only seen by their owners.
Each player remembers the cards that have already been played, including the ones that were used up before the initial state of the $k$-endplay.
That is, the local state of a player contains:
the current hand of the player,
the current hand of the dummy,
the cards from the deck that were already used up in the previous tricks,
the status of the current trick, i.e., the sequence of pairs \textit{(player,card)} for the cards already played within the trick (alternatively, the sequence of cards already played within the trick, plus who started the trick); and
the current score (which team has won how many tricks so far).
\extended{
  \begin{enumerate2}
  \item The current hand of the player,
  \item The current hand of the dummy,
  \item\label{it:deck} The cards from the deck that were already used up in the previous tricks,
  \item The status of the current trick, i.e., the sequence of pairs (player,card) for the cards already played within the trick (alternatively, the sequence of cards already played within the trick, plus who started the trick);
  \item The current score (which team has won how many tricks so far).
  \end{enumerate2}
  Note: since we will only look at the strategic abilities of the declarer (\South), the epistemic relations of the other players (\North, \West, \East) are irrelevant for our experiments.

}

We observe the following properties of the model. First, it is turn-based (with the ``idle'' action \textit{wait} that players use when another player is laying down a card).
Secondly, players have imperfect information, since they cannot infer (except for the last round) the hands of the other players.
The missing information is relevant: anybody who has ever played bridge or poker knows how much the limited knowledge of the opponents' hands decreases one's chances of winning the game.
Thirdly, this is a model of imperfect recall. The players do not remember in which order the cards have been played so far, and who had what cards;\footnote{
  This reflects the capabilities of middle-level bridge players: they usually remember what has been played, but not in which order and by whom. Advanced players remember also who played what, and masters remember the whole history of the play. }
formally: the model is a DAG and not a tree as there are histories $h\not\approx_a h'$ such that $last(h)\sim_a last(h')$).
Finally, the model is lockstep (everybody sees when a transition happens), and thus $\trlow_2$ and $\trlow_3$ coincide on singleton coalitions.

\begin{figure}[h]
\hspace{-0.6cm}\resizebox{1.15\columnwidth}{!}{%
\begin{tabular}{|@{\,}c@{\,}|HHHc|Hc|Hc|Hc|Hc|Hc|c|c|}
\hline
\multirow{2}{*}{$(n,k)$} & \multirow{2}{*}{btime}    & \multirow{2}{*}{\#bsat} & \multirow{2}{*}{rtime}    & \multirow{2}{*}{\#states}  & \multirow{2}{*}{eptime}   & \multirow{2}{*}{tgen}
  & \multicolumn{4}{c|}{Lower approx.\extended{ (\AEMC)}}           & \multicolumn{4}{c|}{Upper approx.\extended{ (\ATLIr)}}      &  \multirow{2}{*}{Match}   &    Exact   \\
\cline{8-15}
\cline{17-17}
  &  &  &  &  &  &  & \#iter & tverif     & \#sat   & \%true & \#iter & tverif   & \#sat   & \%true &  &  tg+tv \\
\hline
\hline
$(1,1)$           & {3.15e-05} & {2}      & {3.29e-04} & {11}      & {1.81e-04} & 0.0005 & 5      & 0.0001 & 11      & 100\%  & 5      & 7e-05 & 11      & 100\%  & {100\%}    & {0.14}   \\ \hline
$(2,2)$           & {9.05e-05} & {6}      & {1.15e-02} & {310}     & {5.37e-03} & 0.017 & 9      & 0.002 & 184     & 60\%   & 9      & 0.001 & 184     & 60\%   & {100\%}    & {2.42\,h}$^\star$
\\ \hline
$(3,3)$           & {3.38e-04} & {20}     & {0.65}     & {12626}   & {0.27}     & 0.92     & 13     & 0.16     & 5672    & 70\%   & 14     & 0.05 & 5929    & 70\%   & {100\%}    & \tmout       \\ \hline
$(4,4)$           & {1.74e-02} & {70}     & {31.22}    & {534722}  & {10.43}    & 41.66    & 19     & 172.07   & 225091  & 60\%   & 19     & 2.61     & 233520  & 60\%   & {100\%}    & \tmout       \\ \hline
$\;(5,5)^\star$           & {5.07e-02} & {252}    & {2449.45}  & {2443467} & {192.36}   & 2641.86  & 15     & 76 h     & 2443467 & 100\%   & 15     & 1929     & 2443467 & 100\%   & {100\%}    & \tmout       \\ \hline
\end{tabular}%
}

\caption{Experimental results: solving endplay in bridge}
\label{fig:resultsbridge}
\end{figure}

The results of the experiments for formula $\phi \equiv \coop[\ir]{\South}\Sometm\prop{win}$ are shown in Figure~\ref{fig:resultsbridge}.
The columns present the following information: parameters of the model $(n,k)$, size of the state space (\#states), generation time for models (tgen), time and output of verification (tver, \%true) for model checking the lower approximation $\trlow(\phi)$, and similarly for the upper approximation $\trup(\phi)$; the percentage of cases where the bounds have matched (match), and the total running time of the exact \ATLir model checking for $\phi$ (tg+tv). The times are given in seconds, except where indicated.

The experiments were run in the same environment as for the voting scenario in Section~\ref{sec:experiments-voting}.
Again, we ran the experiments for up to 48h per instance.
The results in each row are averaged over 20 randomly generated instances, except for ($\star$) where only 1\extended{ hand-picked???} instance was used\extended{, and ($\diamond$) where 2 instances were used}.

\extended{
  The column headers are interpreted as follows:
  \begin{itemize}
  \item btime: creation time of model beginning states,
  \item \#bsat: number of beginning states in model,
  \item rtime: creation time of rest states of model,
  \item \#msat: number of all states in model,
  \item eptime: creation time of epistemic relation in model,
  \item stime: creation time of whole model with epistemic relation,
  \item \#iter: number of iterations until reaching fixpoint,
  \item time: model checking time (in seconds),
  \item \#sat: number of states in which the formula is satisfied,
  \item result: formula result,
  \item match: percentage of times, when imperfect formula result was equal to perfect formula result,
  \item mcmas: time used by mcmas to cumpute answer under imperfect information. \textbf{tmout} denotes that MCMAS did not finish the computation (we ran the experiments with MCMAS for up to 48h per instance).
  \end{itemize}
} 

\para{Discussion of results.}
In the experiments, our approximations offered a dramatic speedup. Exact model checking of $\phi$ was infeasible except for the simplest models (hundreds of states), even with an optimized symbolic model checker like MCMAS. In contrast, the bounds were verified for models up to millions of states.
Moreover, our approximations obtained an astonishing level of accuracy: the bounds matched in 100\% of the analyzed instances, thus producing fully conclusive output.
This was partly because we only considered endplays in relatively small decks. The gap grows for decks of more than 20 cards (we verified that by hand on selected instances from bridge literature).

\extended{
  \begin{figure}[t]
  \hspace{-0.6cm}\resizebox{1.15\columnwidth}{!}{%
\begin{tabular}{|c|HHHc|Hc|Hc|Hc|Hc|Hc|c|H|}
\hline
\multirow{2}{*}{$(n,k)$} & \multirow{2}{*}{btime}    & \multirow{2}{*}{\#bsat} & \multirow{2}{*}{rtime}    & \multirow{2}{*}{\#states}  & \multirow{2}{*}{eptime}   & \multirow{2}{*}{tgen}
  & \multicolumn{4}{c|}{Lower approx.\extended{ (\AEMC)}}           & \multicolumn{4}{c|}{Upper approx.\extended{ (\ATLIr)}}      &  \multirow{2}{*}{Match}   &    Exact   \\
\cline{8-15}
\cline{17-17}
  &  &  &  &  &  &  & \#iter & tverif     & \#sat   & result & \#iter & tverif   & \#sat   & result &  &  tg+tv \\
\hline
\hline
$(1,1,1)$           &  &  &  & {508}    &  & 0.47  &  & 0.006   &  & False   &  & 0.19   &  & False  & {100\%}   &            \\ \hline
$(2,1,1)$           &  &  &  & {1019}   &  & 3.87  &  & 1.91    &  & False   &  & 4.83   &  & True   & {0\%}     &            \\ \hline
$(3,1,1)$           &  &  &  & {2042}   &  & 34.65 &  & 49.81   &  & False   &  & 126.88 &  & True   & {0\%}     &            \\ \hline
\end{tabular}%
}

  \caption{Experimental results for castles, formula castle3defeated}
  \label{fig:resultscastles}
  \end{figure}
} 

\begin{figure}[h]
\hspace{-0.6cm}\resizebox{1.15\columnwidth}{!}{%
\begin{tabular}{|c|HHHc|Hc|Hc|Hc|Hc|Hc|c|c|}
\hline
\multirow{2}{*}{$(n,k)$} & \multirow{2}{*}{btime}    & \multirow{2}{*}{\#bsat} & \multirow{2}{*}{rtime}    & \multirow{2}{*}{\#states}  & \multirow{2}{*}{eptime}   & \multirow{2}{*}{tgen}
  & \multicolumn{4}{c|}{Lower approx.\extended{ (\AEMC)}}           & \multicolumn{4}{c|}{Upper approx.\extended{ (\ATLIr)}}      &  \multirow{2}{*}{Match}   &    Exact   \\
\cline{8-15}
\cline{17-17}
  &  &  &  &  &  &  & \#iter & tverif     & \#sat   & \%true & \#iter & tverif   & \#sat   & \%true &  &  tg+tv \\
\hline
\hline
$(1,1)$           & {5.24e-05} & {2}      & {1.16e-03} & {19}     & {9.89e-06} & 0.001 & 2      & 0.0003 & 19    & 100\%  & 4      & 0.0003 & 19    & 100\%   & {100\%}    & {14.93\,h}$^\star$ \\ \hline
$(2,2)$           & {1.59e-04} & {6}      & {6.78e-02} & {774}    & {1.75e-04} & 0.07 & 2      & 0.01 & 383   & 40\%   & 6      & 0.02 & 410   & 50.00\% & {90\%}     & \tmout        \\ \hline
$(3,3)$           & {7.86e-04} & {20}     & {6.70}     & {51865}  & {9.85e-03} & 6.71     & 5      & 29.31    & 20757 & 65\%   & 11     & 2.45     & 23818 & 85\%    & {80\%}     & \tmout        \\ \hline
\end{tabular}%
}

\caption{Experimental results for absent-minded declarer}
\label{fig:resultsbridge-absentminded}
\end{figure}

\begin{figure}[H]
\hspace{-0.7cm}\resizebox{1.15\columnwidth}{!}{%
\begin{tabular}{|c|HHHc|Hc|Hc|Hc|Hc|Hc|c|c|}
\hline
\multirow{2}{*}{$(n,k)$} & \multirow{2}{*}{btime}    & \multirow{2}{*}{\#bsat} & \multirow{2}{*}{rtime}    & \multirow{2}{*}{\#states}  & \multirow{2}{*}{eptime}   & \multirow{2}{*}{tgen}
  & \multicolumn{4}{c|}{Lower approx.\extended{ (\AEMC)}}           & \multicolumn{4}{c|}{Upper approx.\extended{ (\ATLIr)}}      &  \multirow{2}{*}{Match}   &    Exact   \\
\cline{8-15}
\cline{17-17}
  &  &  &  &  &  &  & \#iter & tverif     & \#sat   & \%true & \#iter & tverif   & \#sat   & \%true &  &  tg+tv \\
\hline
\hline
$(1,1)$           &  &  &  & {19}     &  & 0.002 &  & 0.0001 &  & 0\%    &  & 0.0003 &  & 100\% & {0\%}   & {14.93\,h}$^\star$ \\ \hline
$(2,2)$           &  &  &  & {756}    &  & 0.08  &  & 0.003  &  & 0\%    &  & 0.03   &  & 95\%  & {5\%}   & \tmout        \\ \hline
$(3,3)$           &  &  &  & {55688}  &  & 9.99  &  & 0.09   &  & 0\%    &  & 2.35   &  & 70\%  & {30\%}  & \tmout        \\ \hline
\end{tabular}%
}

\caption{Absent-minded declarer, approximation $\trlow_2$}
\label{fig:resultsbridge-absentminded-classic}
\end{figure}

\subsection{Bridge Endplay by Absentminded Declarer}\label{sec:absentminded}

\wj{
* bridge with absent minded declarer (describe first!):
  not turn-based,
  not lockstep,
  weaker observational capabilities (far from perfect recall!),
  more transitions, as the SN cards can be played in any order, also simultaneously with the other players,
  overall more sophisticated play with (almost) the same state space.
}

In the bridge endplay models, the players always see when a move is made. Thus, for singleton coalitions, the steadfast next-time operator $\diamStead{a}$ coincides with the standard next-time abilities expressed by $\diam{a}$. In order to better assess the performance\extended{ of our lower bound}, we have considered a variant of the scenario where the declarer is absentminded and does not see the cards being laid on the table until the end of each trick. Moreover, she can play her and the dummy's cards at any moment, even in parallel with\extended{ one of} the opponents. This results in larger indistinguishability classes for \South, but also in a general increase of the number of states and transitions\extended{ in the model}.

The results of the experiments are shown in Figure~\ref{fig:resultsbridge-absentminded}. Note that, for this class of models, the bounds do not match as tightly as before. Still, the approximation was conclusive in an overwhelming majority of instances. Moreover, it grossly outperformed the exact model checking which was (barely) possible only in the trivial case of $n=1$.

The models are not turn-based, not lockstep, and not of perfect recall. Since they are not lockstep, approximations $\trlow_2$ and $\trlow_3$ do not have to coincide. In Figure~\ref{fig:resultsbridge-absentminded-classic}, we present the experimental results obtained with $\trlow_2$, which show that the improved approximation $\trlow_3$ provides tighter lower bounds also from the practical point of view.

\wj{Compare the approximations with and without steadfast next step}

\section{Conclusions}

\extended{
  Synthesis of winning strategies for imperfect information games is well known to be hard.
  Similarly, verification
} 
\short{Verification }of strategic properties in scenarios with imperfect information is difficult, both theoretically and in practice.
In this paper, we suggest that model checking of\extended{ strategic} logics like \ATLir can be in some cases obtained by computing an under- and an overapproximation of the \ATLir specification, and comparing if the bounds match.
In a way, our proposal is similar to the idea of may/must abstraction~\cite{Godefroid02abstraction,Ball06abstraction,Lomuscio16abstraction-atlk}, only 
{in our case the approximations are obtained by transforming formulae rather than models}.

We propose such approximations, prove their correctness, and show that, for singleton coalitions, their values can be computed in polynomial time.
We also propose novel benchmarks for experimental validation\extended{, designed so that they share characteristics with simple security scenarios}. Finally, we report very promising experimental results, in both performance and accuracy of the output.
To our best knowledge, this is the first successful attempt at approximating strategic abilities under imperfect information by means of fixpoint methods.

\smallskip\noindent\textbf{Acknowledgements.}
The authors acknowledge the support of the National Centre for Research and Development (NCBR), Poland, under the PolLux project VoteVerif (POL\-LUX-IV/1/2016).


\bibliographystyle{abbrv}
\bibliography{muapprox,wojtek,wojtek-own}


\end{document}